%% file: capest_infocom2011.tex
\documentclass[conference]{IEEEtran}

\makeatletter
\def\ps@headings{%
\def\@oddhead{\mbox{}\scriptsize\rightmark \hfil \thepage}%
\def\@evenhead{\scriptsize\thepage \hfil \leftmark\mbox{}}%
\def\@oddfoot{}%
\def\@evenfoot{}}
\makeatother
\newtheorem{theorem}{Theorem}

\usepackage{cite}
\usepackage{graphicx}
\usepackage{subfigure} 
\usepackage{amsmath}
\usepackage{xspace}
\usepackage{amssymb}

\newcommand{\etal}{\emph{et al.}\xspace}

\begin{document}
\title{CapEst: A Measurement-based Approach to Estimating Link Capacity in Wireless Networks}


\author{\IEEEauthorblockN{Apoorva Jindal}
\IEEEauthorblockA{University of Michigan\\
Ann Arbor, MI 48103 \\
Email: apoorvaj@umich.edu}
\and
\IEEEauthorblockN{Konstantinos Psounis}
\IEEEauthorblockA{University of Southern California\\
Los Angeles, CA 90089\\
Email: kpsounis@usc.edu}
\and
\IEEEauthorblockN{Mingyan Liu}
\IEEEauthorblockA{University of Michigan\\
Ann Arbor, MI 48103\\
Email: mingyan@eecs.umich.edu}
}


%Telephone: (800) 555--1212, Fax: (888) 555--1212}


\maketitle

\begin{abstract}
\input{abstract}
\end{abstract}

\IEEEpeerreviewmaketitle

\input{introduction}

\input{mechanism}
\input{analysis}
\input{sims}

\input{conclusions}

%\section*{Acknowledgment}

\small
\bibliographystyle{IEEEtran}
\bibliography{my}  
\normalsize
\input{appendix}
%\bibliographystyle{IEEEtran}

\end{document}

%% file: abstract.tex
Estimating link capacity in a wireless network is a complex task because the available
capacity at a link is a function of not only the current arrival rate at that link, 
but also of the arrival rate at links which interfere with that link as well as of the nature of interference between these links.
Models which accurately characterize this dependence are either too computationally complex to be useful
or lack accuracy. Further, they have a high implementation overhead and make restrictive assumptions, which makes them inapplicable to real networks.

In this paper, we propose CapEst, a general, simple yet accurate, measurement-based approach to estimating link capacity in a wireless network.
To be computationally light, CapEst allows inaccuracy in estimation; however, using measurements, it can correct this inaccuracy 
in an iterative fashion and converge to the correct estimate. 
Our evaluation shows that CapEst always converged to within $5\%$ of the correct value in less than $18$ iterations.
CapEst is model-independent, hence, is applicable to any MAC/PHY layer and works with auto-rate adaptation. Moreover, it
has a low implementation overhead, can be used with any application which requires an estimate of residual capacity on a wireless link
and can be implemented completely at the network layer without any support from the underlying chipset.

%% file: introduction.tex
\section{Introduction}

The capacity of a wireless link is defined to be the maximum sustainable data arrival rate at that link.
Estimating the residual capacity of a wireless link is an important problem
because knowledge of available capacity is needed by several different tools and applications, including 
predicting safe sending rates of various flows based on policy and path capacity~\cite{qiu:sigcomm},
online optimization of wireless mesh networks using centralized rate control~\cite{ramesh:conext},
distributed rate control mechanisms which provide explicit and precise rate feedback to sources~\cite{our:mobicom}, 
admission control, interference-aware routing~\cite{sigmetrics:802.11}, network management tools to
predict the impact of configuration changes~\cite{qiu:sigcomm} etc. 

However, estimating residual link capacity in a wireless network, especially a multi-hop network, is a hard problem because the available
capacity is a function of not only the current arrival rate at the link under consideration, 
but also of the arrival rates at links which interfere with that link and the underlying topology. 
Models which accurately represent this dependence are very complex and computationally heavy and, as input, require the complete topology
information including which pair of links interfere with each other, the capture and deferral probabilities between each pair of links, 
the loss probability at each link, etc~\cite{sigmetrics:802.11,apoorva:ton,mobicom:complete,qiu:sigcomm,infocom:802.11}. 
Simpler models make simplifying assumptions~\cite{ramesh:conext,Thiran:802.11,kashyap:802.11,reis_sigcomm06,wlan:infocom} 
which diminish their accuracy in real networks. Moreover, model-based
capacity estimation techniques~\cite{qiu:sigcomm,ramesh:conext} work only
for the specific MAC/PHY layer for which they were designed and extending them to a new MAC/PHY layer requires building
a new model from scratch. Finally, none of these methods work with auto-rate adaptation at the MAC layer, which makes them 
inapplicable to any real network.

In this paper, we propose CapEst, a general, simple yet accurate and model-independent, measurement-based approach to estimating link capacity in a wireless network. 
CapEst is an iterative mechanism. During each iteration, 
each link maintains an estimate of the expected service time per packet on that link, and uses this estimate to predict the residual capacity on that link.
This residual capacity estimate may be inaccurate. However, CapEst will progressively improve its estimate with each iteration, and eventually converge to the correct capacity value. 
Our evaluation of CapEst in Sections~\ref{sec:convergence} and~\ref{sec:sims}
shows that, first, CapEst converges, and secondly, CapEst converges to within $5\%$ the correct estimate in less than $18$ iterations.
Note that one iteration involves the exchange of $200$ packets on each link, which lasts less than $280$ms for a link suffering from no interference.
(See Section~\ref{sec:sims} for more details.)

Based on the residual capacity estimate, CapEst predicts the constraints imposed by the network on the rate changes at other links.
CapEst can be used with any application because the application can use it to predict the residual capacity estimate at each link,
and behave accordingly. By a similar argument, CapEst can be used by any network operation to properly allocate resources.
As a show-case, in this paper, we use CapEst for online optimization of wireless mesh networks using centralized rate control.
Later, Section~\ref{sec:applications} briefly illustrates how CapEst is used with two other applications: distributed rate allocation and interference-aware routing.
%Section~\ref{} introduces the CapEst mechanism,
%Section~\ref{} studies the convergence of CapEst theoretically for a simple WLAN topology where all nodes can hear everyone else,
%Section~\ref{} studies the convergence of CapEst through simulations using online optimization of wireless mesh networks as the application,
%and finally, before concluding in Section~\ref{}, Section~\ref{} illustrates how CapEst can be applied to two other applications.

The properties which makes CapEst unique, general and useful in many scenarios are as follows.
(i) It is simple and requires no complex computations, and yet yields accurate estimates.
(ii) It is model-independent, hence, can be applied to any MAC/PHY layer. 
(iii) The only topology information it requires is which node interferes with whom, which can be easily collected locally with low overhead~\cite{ramesh:conext}.
(iv) It works with auto-rate adaptation. 
(v) It can be completely implemented at the network layer and requires no additional support from the chipset.
(vi) CapEst can be used with any application which requires an estimate of wireless link capacity, and on any wireless network, whether single-hop or multi-hop.

\section{Related Work}
\label{sec:related}

\noindent {\bf Model-based Capacity Estimation:} 
Several researchers have proposed models for IEEE 802.11 capacity/throughput estimation in multi-hop networks~\cite{sigmetrics:802.11,apoorva:ton,mobicom:complete,qiu:sigcomm,infocom:802.11,reis_sigcomm06,tobagi:infocom,kashyap:802.11,wlan:infocom,Thiran:802.11,wang:infocom,ramesh:conext}.
%These models tend to be either very complex~\cite{sigmetrics:802.11,apoorva:ton,mobicom:complete,qiu:sigcomm,infocom:802.11}, 
%or make simplifying assumptions on traffic~\cite{reis_sigcomm06,tobagi:infocom}, topology~\cite{kashyap:802.11,wlan:infocom} 
%or the MAC layer~\cite{Thiran:802.11,wang:infocom} or sacrifice on accuracy~\cite{ramesh:conext}.
% 
Model-based capacity estimation suffers from numerous disadvantages. 
They are tied to the model being used, hence, cannot be extended easily to other scenarios, like a different MAC/PHY layer.
Moreover, incorporating auto-rate adaptation would make the models prohibitively complicated, so these techniques do not work with this MAC feature.
Finally, and most importantly, the models used tend to be either very complex and require complete topological information~\cite{sigmetrics:802.11,apoorva:ton,mobicom:complete,qiu:sigcomm,infocom:802.11}, 
or make simplifying assumptions on traffic~\cite{reis_sigcomm06,tobagi:infocom}, topology~\cite{kashyap:802.11,wlan:infocom} 
or the MAC layer~\cite{Thiran:802.11,wang:infocom,ramesh:conext} which significantly reduces their accuracy in a real network.
%To summarize, this approach is either computationally heavy with high overhead or very inaccurate.
%they suffer from the trade-off between accuracy and complexity. Either they are highly complex or very inaccurate.  
%The two most recent models which represent this trade-off, and are closely related to our study as they focus on applying capacity estimation to allocate flow-rates in a centralized manner,
%are~\cite{qiu:sigcomm} and~\cite{ramesh:conext}.~\cite{qiu:sigcomm} uses a complex model which requires the complete topological information like 
%the capture and deferral probabilities between each pair of links etc while~\cite{ramesh:conext} uses a much simpler linear model which requires only 
%the information on which pairs of node interfere, but sacrifices on accuracy. 
Finally, there exists a body of work which studies throughput prediction in wireless networks with optimal scheduling~\cite{jain:mobicom}.
If these are applied directly to predicting rates with any other MAC protocol, it will lead to considerably overestimating capacity. 

\noindent {\bf Congestion/Rate Control:} Numerous congestion control approaches have been proposed to regulate rates at the transport layer,
for example, see~\cite{ATP,our:mobicom,nred,warrier:infocom,congestion:journal,stolyarinfo08} and references therein. These approaches either use message exchanges between interfering nodes
or use back-pressure. All these approaches improve fairness/throughput; however, they do not allow the specification of a well-defined throughput fairness
objective (like max-min or proportional fairness). Moreover, most of them require changes to either IEEE 802.11 MAC or TCP. Our work is complementary to these works. 
We focus on how to accurately estimate capacity irrespective of the MAC or the transport layer. 
This capacity estimate, amongst many applications, can also be used to build a rate control protocol which allows
the intermediate routing nodes to provide explicit and precise rate feedback to source nodes (like XCP or RCP in wired networks).
It can also be used to set long-term rates while retaining TCP, similar to the model-based capacity estimation
techniques of~\cite{qiu:sigcomm} and~\cite{ramesh:conext}, to get better end-to-end throughputs.

\noindent {\bf Asymptotic Capacity Bounds:} An orthogonal body of work studies asymptotic capacity bounds in multi-hop networks~\cite{Kumar:capacity,Tse:Mobility,vetterli:relay}.
While these works lend useful insight into the performance of wireless networks in the limit, their models are abstract by necessity and
cannot be used to estimate capacity in any specific real network.

%% file: mechanism.tex
\section{{CapEst} Description}
\label{sec:capest}

CapEst is an iterative mechanism. During each iteration, each link measures the expected service time per packet on that link,
and uses this measurement to estimate the residual capacity on that link. The application, which for this section is centralized rate 
allocation to get max-min fairness, uses the estimated residual capacity to allocate flow-rates. 
The estimate may be inaccurate, however, it will progressively improve with each iteration, and eventually converge to the correct value.
This section describes each component of CapEst in detail and the max-min fair centralized rate allocator.

%Notation - define neighborhood, arrival rate and expected service time
We first define our notations. Let $V$ denote the set of wireless nodes. A link $i \rightarrow j$ is described by the 
transmitter-receiver pair of nodes $i,j \in V, i \neq j$. Let $\mathcal{L}$ denote the set of active links. 
Let $\lambda_{i \rightarrow j}$ denote the packet arrival rate at link $i \rightarrow j$.
Let $N_{i \rightarrow j}$ denote the set of links which interfere with $i \rightarrow j$, where a link $k \rightarrow l$ is defined to interfere
with link $i \rightarrow j$ if and only if either $i$ interferes with node $k$ or $l$, or $j$ interferes with node $k$ or $l$. 
For convenience, $i \rightarrow j \in N_{i \rightarrow j}$. ($N_{i \rightarrow j}$ is also referred to as the neighborhood of $i \rightarrow j$~\cite{our:mobicom}.) 
%Finally, let $E\left[ S_{i \rightarrow j} \right]$
%denote the expected service time a packet sees at the medium access layer of link $i \rightarrow j$. In other words, $E\left[ S_{i \rightarrow j} \right]$
%is expected duration between sending the packet from the network layer to the MAC layer and receiving a notification from the MAC
%that the packet has been successfully transmitted.

In this description, we make the following assumptions.
(i) The retransmit limit at the MAC layer is very large, so no packet is dropped by the MAC layer. 
(ii) The size of all packets is the same and the data rate at all links is the same.
Note that these assumptions are being made for ease
of presentation, and in Sections~\ref{sec:finite} and~\ref{sec:auto-rate}, we present modifications to CapEst to incorporate finite retransmit limits and different
packet sizes and data rates respectively.

\subsection{Estimating Capacity}
%We now describe the basic CapEst mechanism. 
%The metric which CapEst uses to estimate link capacity is the expected service time at that link.
%
We first describe how each link estimates its expected service time. 
For each successful packet transmission, CapEst measures the time elapsed between the MAC layer receiving the packet 
from the network layer, and the MAC layer informing the network layer that the packet has been successfully transmitted. This denotes the service time of that packet.
Thus, CapEst is completely implemented in the network layer.
%, and does not require any support from the chipset implementing the MAC and PHY layers. 
CapEst maintains the value of two variables $\overline{S}_{i \rightarrow j}$ and $K_{i \rightarrow j}$\footnote{We could have used an 
exponentially weighted moving average to estimate $\overline{S}_{i \rightarrow j}$ too. However, since each packet is served by the same service process, giving equal weight to
each packet in determining $\overline{S}_{i \rightarrow j}$ should yield better estimates, which we verify through simulations.}, which denote the estimated expected service time and a counter
to indicate the number of packets over which the averaging is being done respectively, at each link $i \rightarrow j$.
For each successful packet transmission, if $S_{i \rightarrow j}^{last}$ denotes the service time of the most recently transmitted packet, 
then the values of these two variables are updated as follows.
\begin{eqnarray} 
\label{eqn:e1} & & \overline{S}_{i \rightarrow j} \leftarrow \frac{\overline{S}_{i \rightarrow j} \times K_{i \rightarrow j} + S_{i \rightarrow j }^{last}}{K_{i \rightarrow j} + 1} \\
\label{eqn:e2} & & K_{i \rightarrow j} \leftarrow K_{i \rightarrow j} + 1. 
\end{eqnarray}

$1/\overline{S}_{i \rightarrow j}$ gives the MAC service rate at link
$i \rightarrow j$. Thus, the residual capacity on link $i \rightarrow j$ is equal to $\left(1/\overline{S}_{i \rightarrow j}\right) - \lambda_{i \rightarrow j}$.
Now, since transmissions on neighboring links will also eat up the capacity at link $i \rightarrow j$, this residual capacity will be
distributed amongst all the links in $N_{i \rightarrow j}$. 
Note that the application using CapEst, based on this residual capacity estimate, will either re-allocate rates or change the routing
or admit/remove flows etc, which will change the rate on the links in the network. 
However, this rate change will have to obey the following constraint at each link to keep the new rates feasible.
\begin{equation}
\label{eqn:constraint}
\sum_{k \rightarrow l \in N_{i \rightarrow j}} \delta_{k \rightarrow l} \leq 
\left( 1/\overline{S}_{i \rightarrow j} \right) - \lambda_{i \rightarrow j}, \forall i \rightarrow j \in \mathcal{L},
\end{equation}
where $\delta_{k \rightarrow l}$ denotes the rate increase at link $k \rightarrow l$ in packets per unit time. How exactly is this residual capacity divided 
amongst the neighboring links depends on the application at hand. For example, 
Section~\ref{sec:max-min-calc} describes a centralized methodology to distribute this estimated residual capacity amongst interfering links so as 
to obtain max-min fairness amongst all flows. Similarly, Section~\ref{sec:weighted} describes how to divide this estimated capacity to get weighted fairness.

Note that all interfering links in $N_{i \rightarrow j}$ do not have the same effect on the link $i \rightarrow j$. 
Some of these interfering links can be scheduled simultaneously, and some do not always interfere and packets may go through
due to capture affect (non-binary interference~\cite{mobicom:complete}). 
However, the linear constraint of Equation (\ref{eqn:constraint}) treats the links in $N_{i \rightarrow j}$
as a clique with binary interference. Hence, there may be some remaining capacity on link $i \rightarrow j$ after utilizing
Equation (\ref{eqn:constraint}) to ensure that the data arrival rates remain feasible.
% it may lead to underestimation of the capacity. 
CapEst will automatically improve the capacity estimate in the next iteration and converge to the correct capacity value iteratively.
%After the residual capacity has been distributed amongst the interfering links, $i \rightarrow j$ will again measure $E\left[ S_{i \rightarrow j} \right]$, and use it again to estimate the residual capacity. And this residual capacity will again be distributed amongst the interfering links. 

%One iteration involves estimating $\overline{S_{i \rightarrow j}}$, and then having the application distribute any remaining residual capacity amongst neighboring links. 
We refer to the duration of one iteration as the {\it iteration duration}. At the start of the iteration, 
both the variables $\overline{S}_{i \rightarrow j}$ and $K_{i \rightarrow j}$ are initialized to $0$. We start 
estimating the expected service time afresh at each iteration because the residual capacity distribution in the previous iteration may
change the link rates at links in $N_{i \rightarrow j}$, and hence change the value of $\overline{S}_{i \rightarrow j}$. 
Thus, retaining $\overline{S}_{i \rightarrow j}$ from the previous iteration is inaccurate.
We next discuss how to set the value of the iteration duration.
Note that CapEst has no overhead; it only requires each link to determine $\overline{S}_{i \rightarrow j}$ which does not require
any message exchange between links.
%If CapEst has no overhead, then does it have any constraint on the choice of the iteration duration? 
This however does not imply that there is no constraint on the choice of the iteration duration.
Each link $i \rightarrow j$ has to measure $\overline{S}_{i \rightarrow j}$ afresh.
Thus, an iteration duration has to be long enough so as to accurately measure $\overline{S}_{i \rightarrow j}$.
However, we cannot make the iteration duration too long as it directly impacts the convergence time of CapEst. 
Moreover, the application which will distribute capacity will have overhead as it will need message exchanges to ensure Equation (\ref{eqn:constraint}) holds.
The choice of the value of the iteration duration is further discussed in Section~\ref{sec:sims}.
%Through simulations, we determine that choosing an iteration duration long enough to exchange $200$ packets on each link is sufficient to capture the
%change in the expected service times. (Section~\ref{sec:sim_interval_length} discusses the impact of choosing a smaller iteration duration.) 
%Based on the convergence time requirements and overhead of the application, one can decide to choose a longer iteration duration too. 

\subsection{Distributing Capacity to Obtain Max-Min Fairness}
\label{sec:max-min-calc}
To illustrate how CapEst will be used with a real application, we now describe a centralized mechanism to allocate max-min fair rates to flows.

Each link determines its residual capacity through CapEst and conveys this information to a centralized rate allocator. 
This centralized allocator is also aware of which pair of links in $\mathcal{L}$ interfere with each other as well as
the routing path of each flow. Let $\mathcal{F}$ denote the set of end-to-end flows characterized by their source-destination pairs.
Let $r_f$ denote the new flow-rate determined by the centralized rate allocator, and let $r_{i \rightarrow j}^{allocate}$ denote the maximum flow rate
allowed on link $i \rightarrow j$ for any flow passing through this link.  

Consider a link $i \rightarrow j$. Let $I(f, i \rightarrow j)$ be an indicator variable which is equal to $1$ only if flow $f \in \mathcal{F}$ passes through a link
$i \rightarrow j \in \mathcal{L}$, otherwise it is equal to $0$. 
Based on the residual capacity estimates, the centralized allocator updates the value of $r_{f}, f \in \mathcal{F}$ according to the following set of equations.
\begin{eqnarray}
& & r_{i \rightarrow j}^{max} \leftarrow r_{i \rightarrow j}^{allocate} + \frac{1/\overline{S}_{i \rightarrow j} - \lambda_{i \rightarrow j}}{\sum_{k \rightarrow l \in N_{i \rightarrow j}} \sum_{f \in \mathcal{F}} I(f, k \rightarrow l)}  \nonumber \\
& & r_{i \rightarrow j}^{allocate} \leftarrow \mbox{min}_{k \rightarrow l \in N_{i \rightarrow j}} r_{k \rightarrow l}^{max}  \nonumber \\
& & r_f \leftarrow \mbox{min}_{i \rightarrow j \in P_f} r_{i \rightarrow j}^{allocate}, 
\end{eqnarray}
where $r_{i \rightarrow j}^{max}$ denotes the maximum flow-rate allowed 
by link $i \rightarrow j$ on any link in $N_{i \rightarrow j}$ and the set $P_f$ contains the links lying on the routing path of flow $f$.
Thus, amongst the links a flow traverses, its rate is updated according to the link with the minimum residual capacity in its neighborhood.
Note that the residual capacity estimate may be negative, which will merely result in reducing the value of $r_{f}$.

This mechanism will allocate equal rates to flows which pass through the neighborhood of the same bottleneck link.
According to the proof presented in~\cite{wcp:ton}, for CSMA-CA based MAC protocols, this property ensures max-min fairness. 

\subsection{Finite Retransmit Limits}
\label{sec:finite}
A packet may be dropped at the MAC layer if the number of retransmissions exceeds the maximum retransmit limit. 
Since this packet was dropped without being serviced, what is its service time? Should we ignore this packet and not change the current estimate
of the expected service time, or, should we merely take the duration for which the packet was in the MAC layer and use it as a measure of its service time?
Note that these lost packets may indicate that the link is suffering from severe interference, and hence, imply that flows passing through the neighborhood 
of this link should reduce their rates. Ignoring lost packets is thus not the correct approach. Moreover, merely using the duration the lost packet spent in the MAC
layer as the packet's service time is not sufficient to increase the value of the expected service time by an amount which leads to a negative residual capacity.

We use the following approach. We define the service time of a lost packet to be equal to the sum of the duration spent by the packet in the MAC layer and the expected 
additional duration required by the MAC layer to service the packet if it was not dropped. Assuming independent losses\footnote{A more complex model
accounting for correlated losses can be easily incorporated, however, our simulation results presented in Section~\ref{sec:sim_finite} show that
this simple independent loss model yields accurate results.}, the latter term is equal to $\frac{W_m/2 + T_s}{1 - p_{i \rightarrow j}^{loss}}$,
where $W_m$ is the largest back-off window value, $T_s$ is the transmission time of a packet and $p_{i \rightarrow j}^{loss}$ is the probability that a 
DATA transmission on link $i \rightarrow j$ is not successful. This value can be directly monitored at the network layer by keeping a running ratio of the
number of packets lost to the number of packets sent to the MAC layer for transmission\footnote{Actually, the correct approach is to use the probability of loss seen at the PHY 
layer, but measuring this quantity requires support from the chipset firmware. If the firmware supports measuring this loss probability, it should be used.}
%TechReport
given that each lost packet is transmitted as many times as the MAC retransmit limit.

\subsection{Differing Packet Sizes and Data Rates}
\label{sec:auto-rate}
The methodology to estimate the expected service time (Equations (\ref{eqn:e1}) and (\ref{eqn:e2})) remains the same. 
What changes is how to distribute this residual capacity amongst neighboring links, which is governed by the constraint of Equation (\ref{eqn:constraint}).
The objective of this constraint is to ensure that the sum of the increase in the proportion of time a link $k \rightarrow l \in N_{i \rightarrow j}$ transmits
should be less than the proportion of time the channel around $i$ and $j$ is empty.
Thus, if each link has a different transmission time, then the increase in rates on $k \rightarrow l \in N_{i \rightarrow j}$ as well as the residual capacity
on $i \rightarrow j$ has to be scaled so as to represent the increase in airtime being consumed and the idle airtime around $i$ and $j$ respectively.
Hence, if $\overline{T}_{i \rightarrow j}$ represents the average packet transmission time at link $i \rightarrow j$, then the capacity
estimation mechanism will impose the following constraint. 
\begin{equation}
\label{eqn:new_constraint}
\sum_{k \rightarrow l \in N_{i \rightarrow j}} \delta_{k \rightarrow l} \frac{\overline{T}_{k \rightarrow l}}{\overline{T}_{i \rightarrow j}} \leq 
\left( 1/\overline{S}_{i \rightarrow j} \right) - \lambda_{i \rightarrow j}, \forall i \rightarrow j \in \mathcal{L}.
\end{equation}
Equation (\ref{eqn:new_constraint}) merely normalizes rates to airtime.
Note that Equation (\ref{eqn:constraint}) is a special case of Equation (\ref{eqn:new_constraint}) if the packet sizes and data rates at all links are the same.

Finally, the following equation states how the value of $\overline{T}_{i \rightarrow j}$ is estimated.
\begin{equation}
\label{eqn:est_t}
\overline{T}_{i \rightarrow j} \leftarrow \frac{\overline{T}_{i \rightarrow j} \times K_{i \rightarrow j} + \frac{P_{i \rightarrow j}^{last}}{D_{i \rightarrow j}^{last}}}{K_{i \rightarrow j} + 1},
\end{equation}
where $P_{i \rightarrow j}^{last}$ and $D_{i \rightarrow j}^{last}$ denote the packet size of the last packet transmitted and the data rate used to transmit
the last packet respectively for link $i \rightarrow j$. 
%Recall that $K_{i \rightarrow j}$ is a counter to maintain the number of packets over which the averaging is being performed.
%\subsection{Properties of CapEst}
%Move this to conclusions??>
%CapEst has numerous desirable properties. We now briefly discuss a few of them. 
%\begin{itemize}
%\item[(i)] CapEst does not depend on the MAC and the PHY layers being used below. Since, it does not use a model to estimate capacity, it 
%makes no assumption on these layers satisfying a given set of equations. As discussed earlier, even though the constraint specified by Equation (\ref{eqn:constraint})
%may lead to under-estimating capacity, CapEst will fix this under-estimation over the next few iterations.
%\item[(ii)] CapEst only requires a knowledge of who interferes with whom so as to be able to impose Equation (\ref{eqn:constraint}).
%It does not need any knowledge of the exact topology in terms of the loss probability at each link and the capture and deferral probabilities
%for each pair of links. As discussed in~\cite{ramesh:conext}, this requires much less overhead.  
%\item[(iii)] CapEst makes no assumption on the application which will use it as a tool to estimate capacity.
%Till the application respects Equation (\ref{eqn:constraint}) in distributing the residual capacity, CapEst converges to the correct solution.
%\item[(iv)] CapEst works with auto-rate adaptation and finite retransmit limits.  
%\item[(v)] Implementation Overhead.
%\item[(vi)] CapEst requires no support from the underlying chipset and it can be implemented completely in software at the network layer. 
%\end{itemize}

%% file: analysis.tex
\section{Convergence of {CapEst}}
\label{sec:convergence}

%To understand the convergence of CapEst analytically while maintaining tractability,
We analyze the convergence of CapEst with the max-min fair centralized rate allocator under a WLAN topology, where all nodes can hear each other, 
all nodes are homogeneous, and each link has the same packet arrival rate. 
Later, through simulations, we show that, with general multi-hop topologies, with fading and shadowing, heterogeneous nodes
and different arrival rates for each link, our results still hold. 
%TechReport
%We use the small buffer model for the WLAN topology proposed by Zhao \etal~\cite{zhao:fixed_point}. 
\begin{theorem}
\label{thm:convergence}
Let $\lambda_{fin}$ denote the arrival rate at each edge such that the expected service time at each edge is equal to $1/\lambda_{fin}$.
For a WLAN topology where all nodes can hear each other, with homogeneous nodes having small buffers, and with each link having the same packet arrival rate, 
$\lambda_{fin}$ always exists and is unique, and CapEst, with the max-min fair centralized rate allocator, 
always converges to $\lambda_{fin}$.
\end{theorem}
The details of the proof are given in the Appendix.

%% file: sims.tex
\section{Performance}
\label{sec:sims}
In this section, we demonstrate through extensive simulations that CapEst not only converges quickly but also converges to the correct 
rate allocation. Thus, we verify both the correctness and the convergence of CapEst.  

\begin{figure*}
\centerline{\subfigure[]{\includegraphics[width=2.5cm]{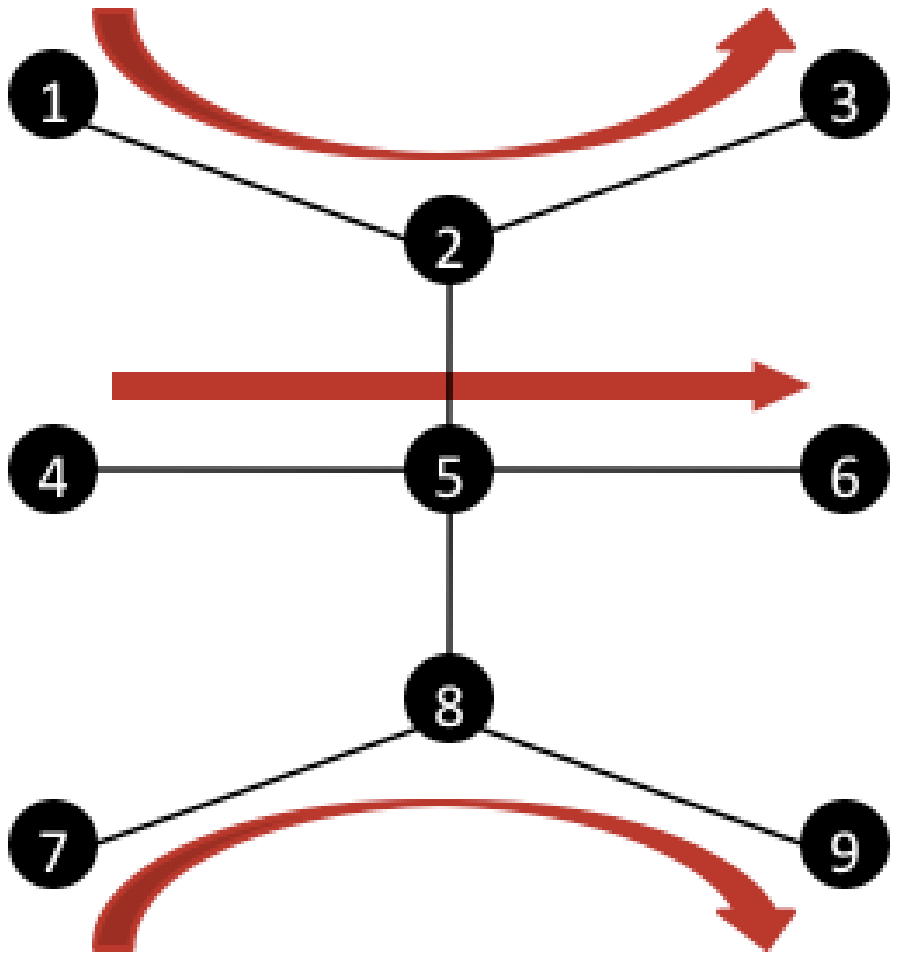}
\label{fig:fim}}
\hfil
\subfigure[]{\includegraphics[width=5.0cm]{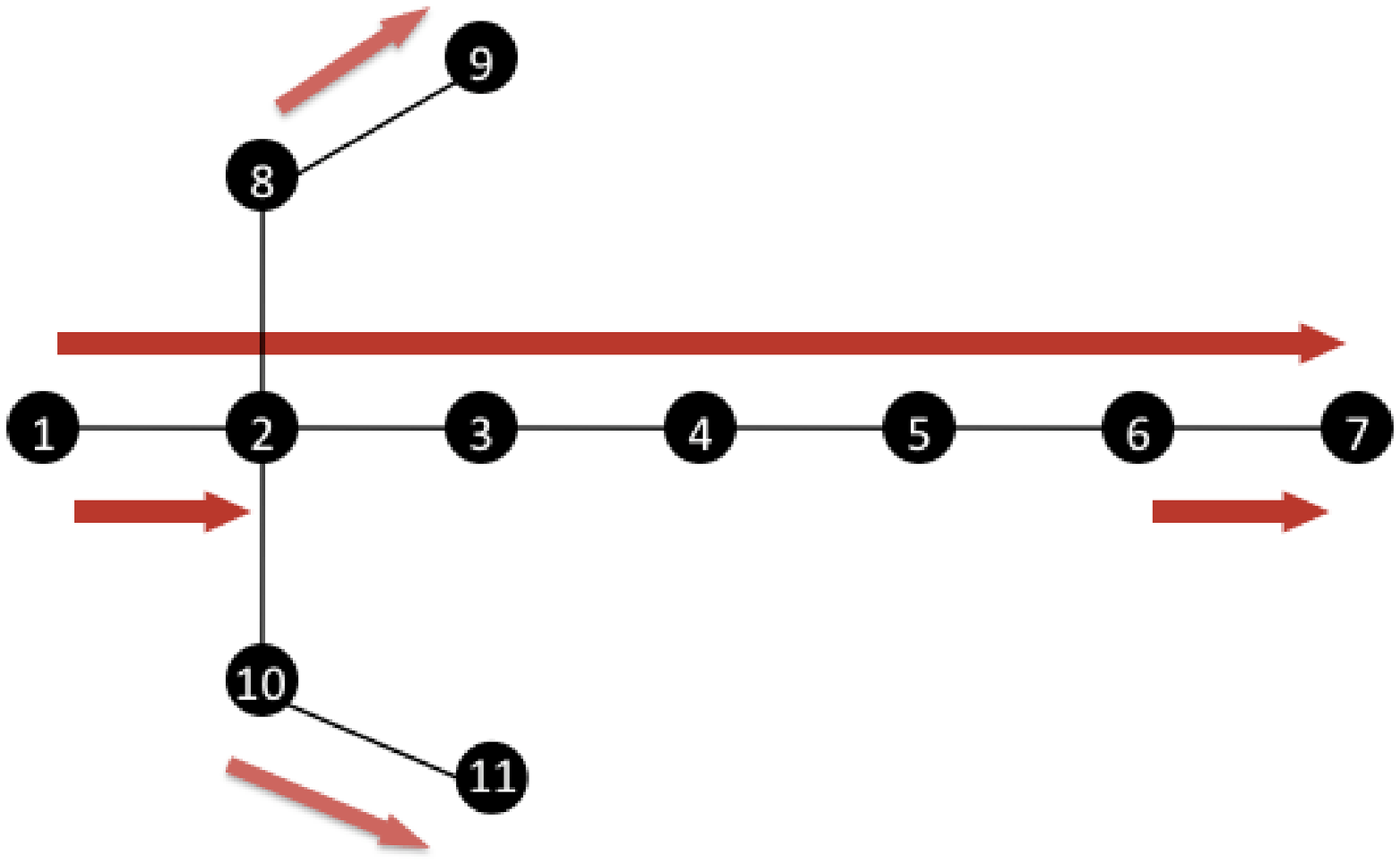}
\label{fig:chain_cross}}
\hfil
\subfigure[]{\includegraphics[width=4.5cm]{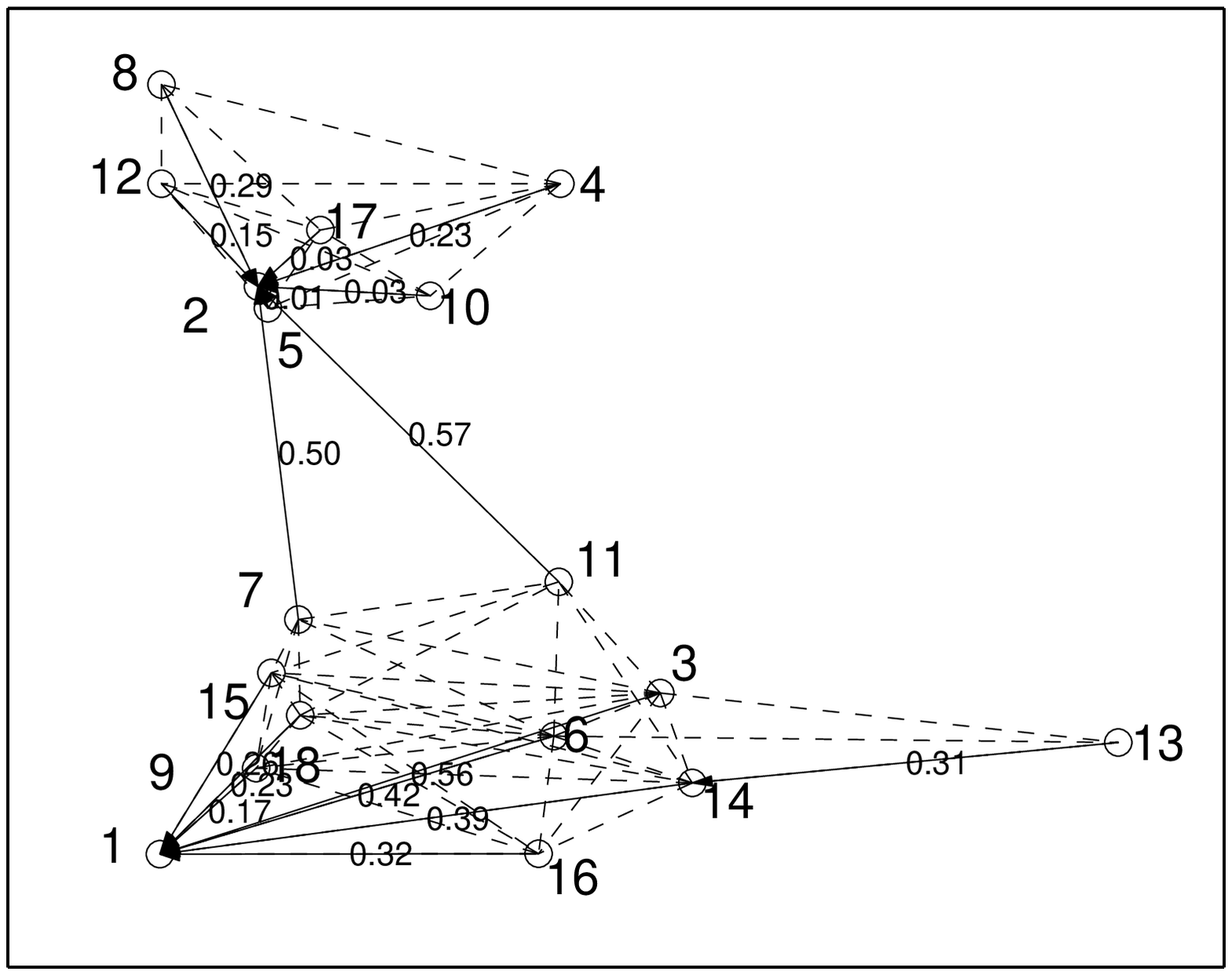}
\label{fig:trice}}
\hfil
\subfigure[]{\includegraphics[width=4.0cm]{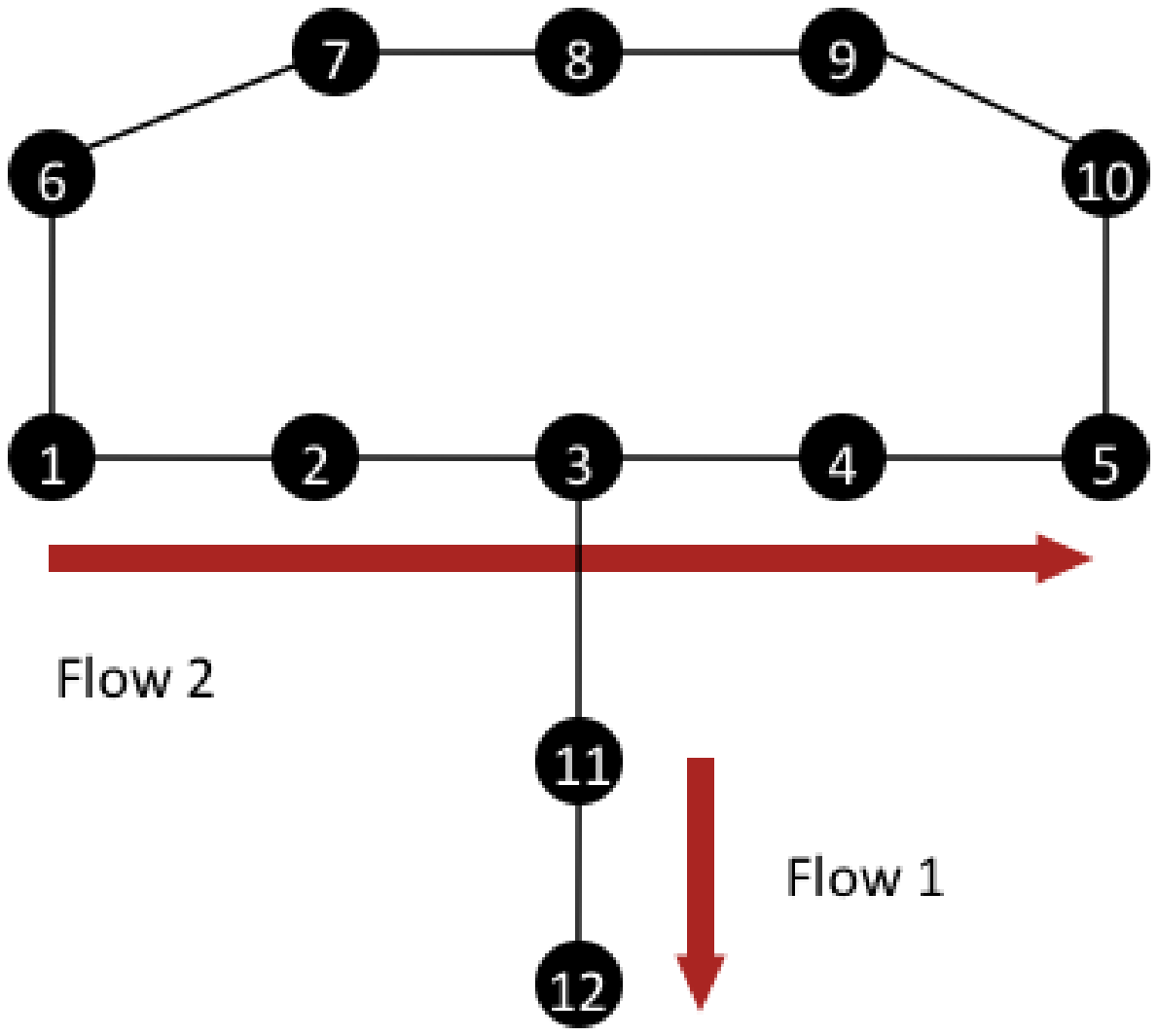}
\label{fig:routing_top}}
}
\caption{(a) Flow in the Middle Topology. (b) Chain-cross Topology. (c) Deployment at Houston. (d) Topology to evaluate interference-aware routing.}
\end{figure*}
%TechReport
\begin{table}
\small
\begin{center}
\begin{tabular}{|c|c|}  \hline
Packet Payload & $1024$ bytes \\ \hline MAC Header &  $34$ bytes \\
\hline PHY Header & $16$ bytes \\ \hline ACK & $14$ bytes + PHY
header \\ \hline RTS & $20$ bytes + PHY header \\ \hline CTS & $14$
bytes + PHY header \\ \hline Channel Bit Rate & $11$ Mbps \\ \hline
Propagation Delay & $1$ $\mu$s \\ \hline Slot Time & $20$ $\mu$s \\
\hline SIFS & $10$ $\mu$s \\ \hline DIFS & $50$ $\mu$s \\ \hline
Initial backoff window & $31$ \\ \hline Maximum backoff window & $1023$ \\ \hline
\end{tabular}
\end{center}
\normalsize
\caption{Simulation parameters.}
\label{parameters}
\end{table}

We evaluate CapEst with the max-min fair centralized rate allocator over a number of different topologies, 
for different MAC protocols, with finite retransmit values at the MAC layer and with auto-rate adaptation.
We observe that CapEst {\it always} converges to within $5\%$ of the optimal rate allocation in less than $18$ iterations.

\subsection{Methodology}
We use Qualnet version 4.0 as the simulation platform. All our simulations are conducted using an unmodified IEEE 802.11(b) MAC (DCF). RTS/CTS 
is not used unless explicitly stated. We use the default parameters of IEEE 802.11(b) 
%TechReport
(summarized in Table~\ref{parameters}) 
in Qualnet. Unless explcitly stated, auto-rate adaptation is turned off, the link rate is set to $11$ Mbps, the packet size is set to $1024$ bytes
and the maximum retransmit limits are set to a very large value. This setting allows us to first evaluate the performance of the basic CapEst mechanism without the
modifications for auto-rate adaptation and finite retransmit limit. Later, we include both 
to evaluate the mechanisms to account for them. We run bulk transfer flows till $10,000$ packets per flow have been delivered.
(Section~\ref{sec:sim_discuss} discusses the impact of having short-flows in the network.)
Finally, to determine the actual max-min rate allocation, we use the methodology proposed by~\cite{our:mobicom}.

The implementation of CapEst and the max-min fair centralized rate allocator closely follows their description in Section~\ref{sec:capest}.
We choose the iteration duration to be $200$ packets, that is, each link resets its estimate of expected service time after transmitting $200$ packets\footnote{Note 
that based on the convergence time requirements and overhead of the application, one can decide to choose a longer iteration duration too.}.
Finally, to be able to correctly distribute capacity, the centralized rate allocator also needs to be aware of which links interfere with each other. 
We use the binary LIR interference model described in~\cite{ramesh:conext} to determine which links interfere. 
The link interference ratio (LIR) is defined as $LIR = \frac{c_{31} + c_{32}}{c_{11}+c_{22}}$ where $c_{11}, c_{22}$ and $c_{31}, c_{32}$
are UDP throughputs when the links are backlogged and transmit individually and simultaneouly respectively. $LIR = 1$ implies no interference,
with lower LIR's indicating a higher degree of interference. Similar to the mechanism in~\cite{ramesh:conext}, links with $LIR>0.95$ are
classified as non-interfering. 

\begin{figure*}
\centerline{\subfigure[]{\includegraphics[width=4.5cm]{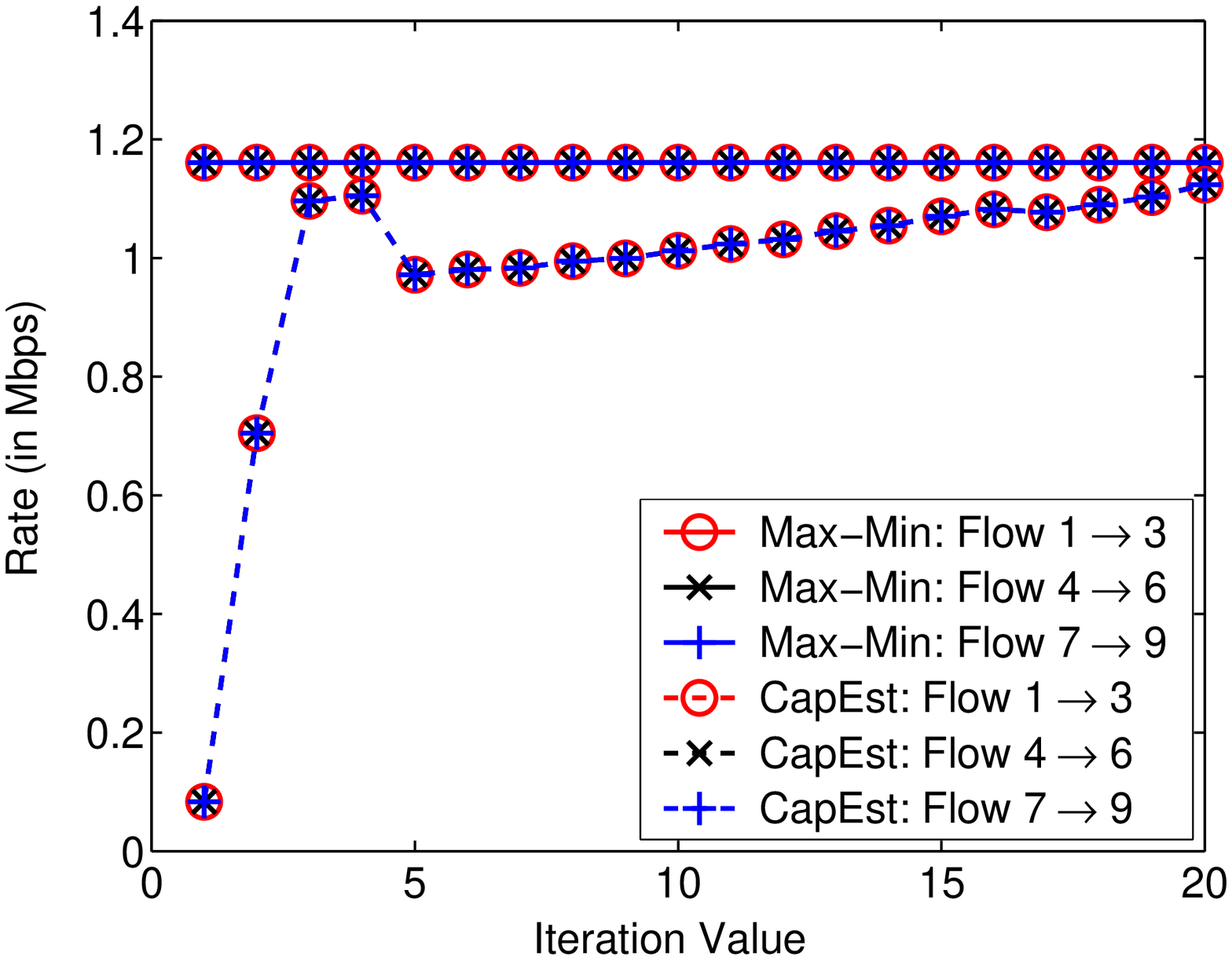}
\label{fig:plot_fim}}
\hfil
\subfigure[]{\includegraphics[width=4.5cm]{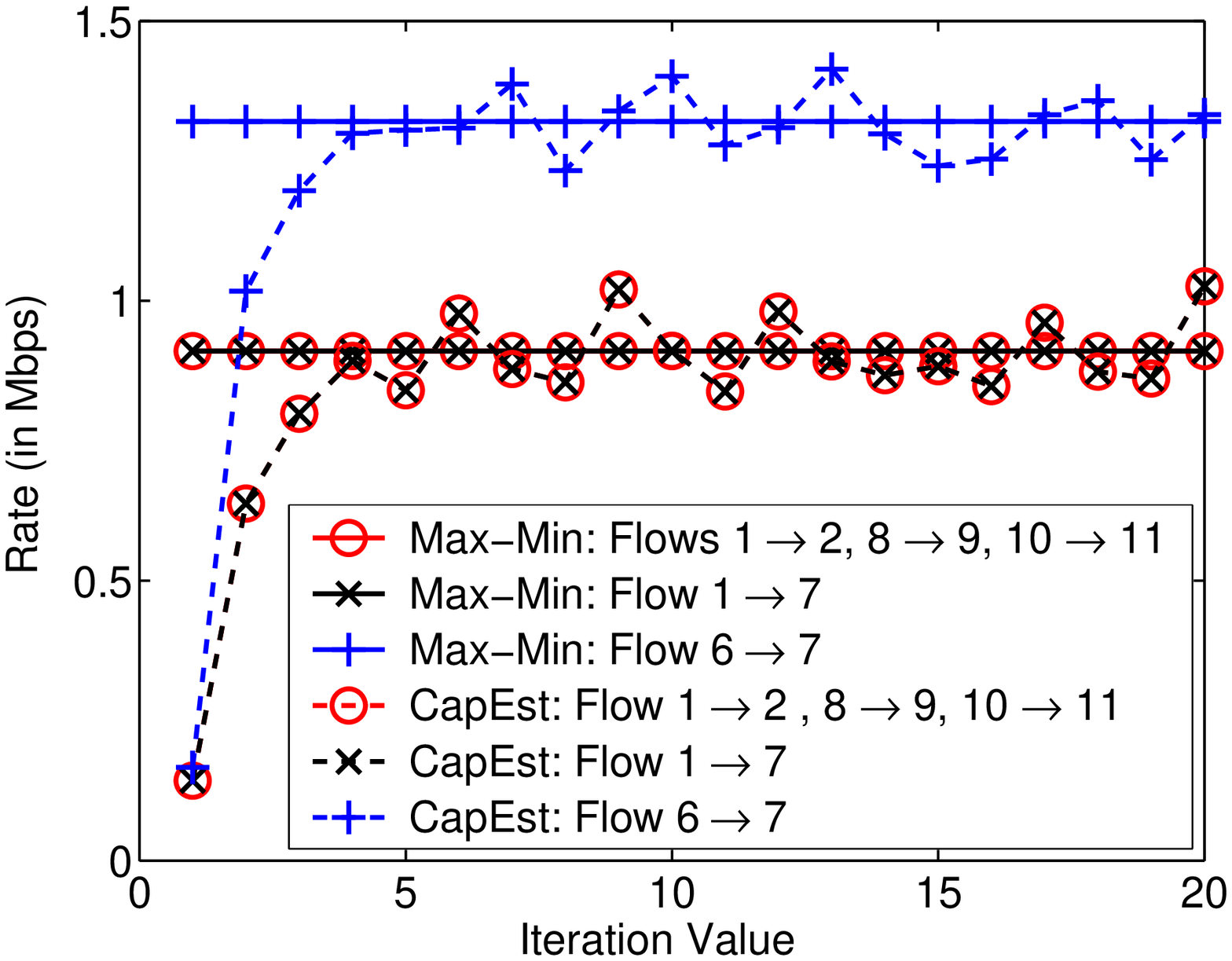}
\label{fig:plot_chain_cross}}
\hfil
\subfigure[]{\includegraphics[width=4.5cm]{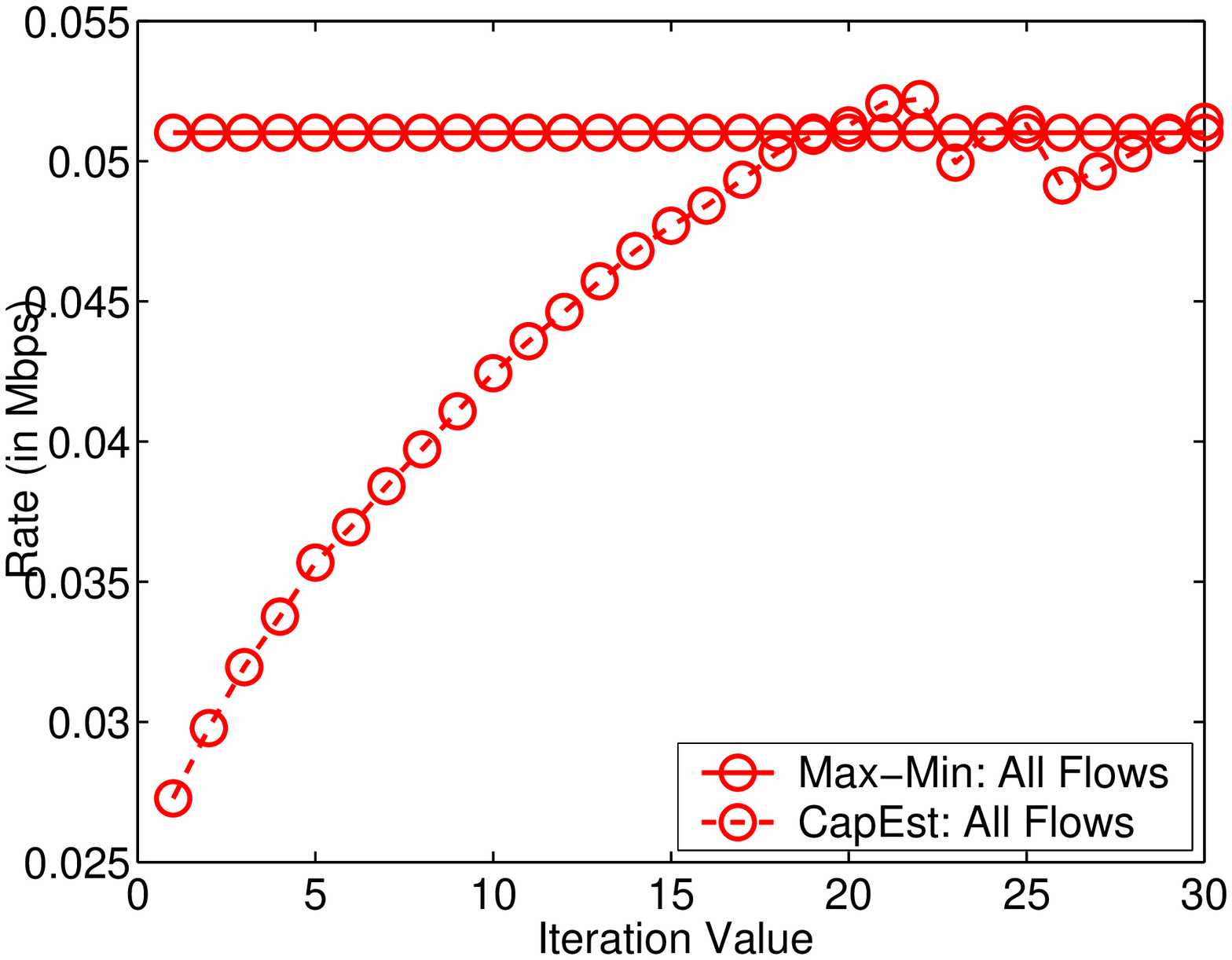}
\label{fig:plot_rnd_small}}
\hfil
\subfigure[]{\includegraphics[width=4.5cm]{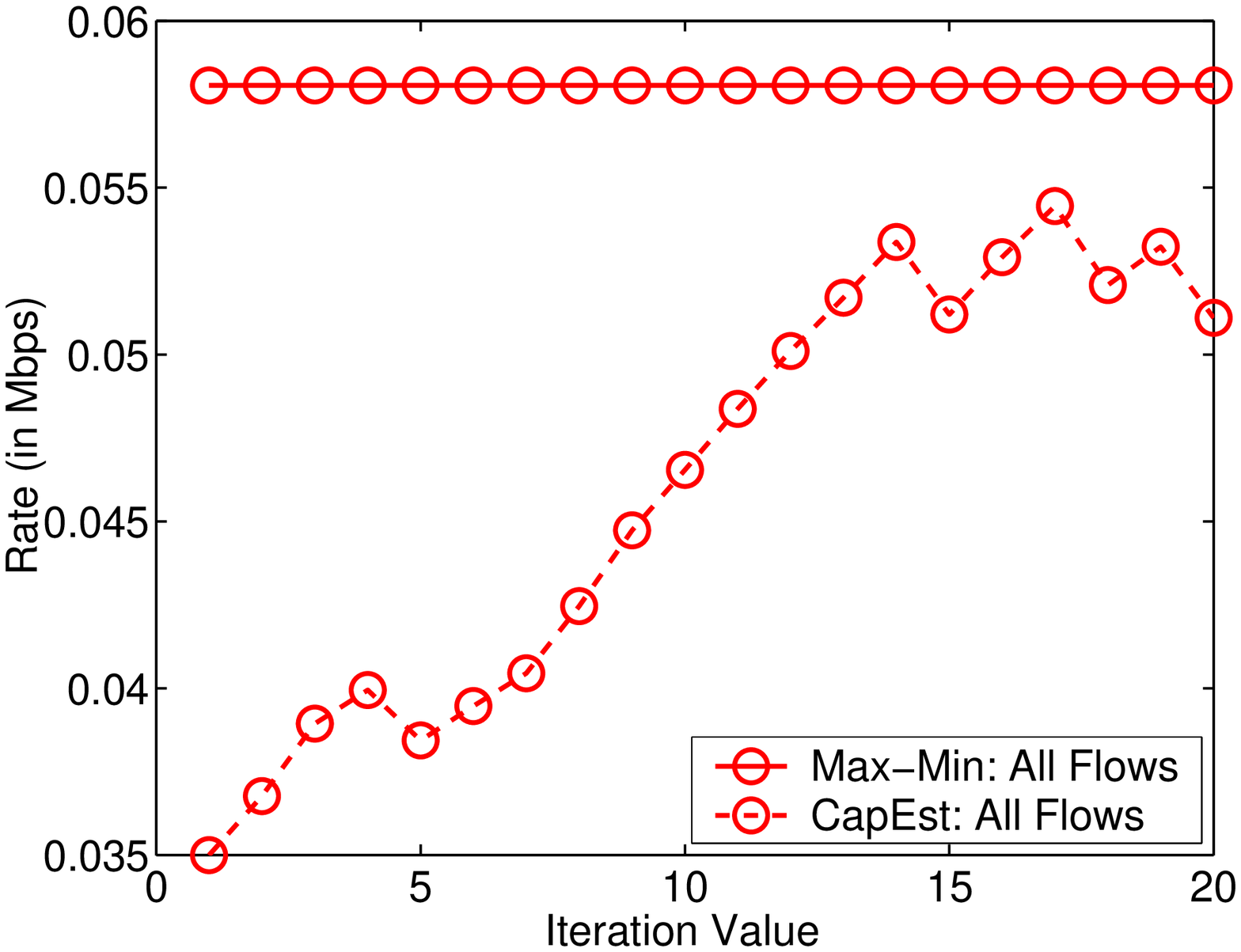}
\label{fig:plot_rnd_large}}}
\caption{Performance of CapEst. (a) Flow in the Middle. (b) Chain-Cross. (c) Random topology: $30$ nodes, $5$ flows. (d) Random topology: $100$ nodes, $10$ flows.}
\end{figure*}

\subsection{Commonly Used Multi-hop Topologies}
In this section, we evaluate CapEst with two different, commonly used topologies. 
Simulations on these two topologies are conducted with zero channel losses, although packet losses due to collisions do occur.
We adjusted the carrier sense threshold to reduce the interference range to be able to generate these topologies.

\subsubsection{Flow in the Middle Topology} Figure~\ref{fig:fim} shows the topology. This topology has been studied and used by several researchers
to understand the performance of different rate control and scheduling protocols in mesh networks~\cite{apoorva:ton,our:mobicom,nred}.
Figure~\ref{fig:plot_fim} plots the evolution of the rate assigned to each flow by the centralized rate allocator.
We observe that the mechanism converges to within $5\%$ of the optimal max-min rate allocation in less than $16$ iterations. 

\subsubsection{Chain-Cross Topology} Figure~\ref{fig:chain_cross} shows the topology. This topology was proposed by~\cite{our:mobicom} to understand the
performance of rate control protocols in mesh networks. This topology has a flow in the middle which goes over multiple hops ($1 \rightarrow 7$) 
as well as a smaller flow in the middle ($1 \rightarrow 2$). Figure~\ref{fig:plot_chain_cross} plots the evolution of the rate assigned
to each flow by the centralized rate allocator. We see that the mechanism converges to within $5\%$ of the optimal in less than $5$ iterations.
Note that the flow-rate allocated slightly exceeds the optimal allocation at certain iteration values. This occurs because of the estimation error
in the measurement of expected service time. Thus, even though Equation (\ref{eqn:constraint}) is set-up to keep the rate allocation feasible,
error in the estimation of the expected service time can lead to infeasible rate allocations.

\subsection{Randomly Generated Topologies}
\label{sec:rnd}
We generate two topologies by distributing nodes in a square area uniformly at random. 
The source-destination pairs are also randomly generated. The first random topology has $30$ nodes and $5$ flows, while
the second has $100$ nodes and $10$ flows. We use the two-ray path loss model with Rayleigh fading and log-normal shadowing~\cite{Rappaport:book} 
as the channel model in simulations. The carrier-sense threshold, the 
noise level and the fading and shadowing parameters are set to their default values in Qualnet. We use AODV to set up the routes. 
Figures~\ref{fig:plot_rnd_small} and~\ref{fig:plot_rnd_large} plot the evolution of the rate assigned to each flow by the centralized rate allocator. 
For the smaller random topology, the mechanism converges to within $5\%$ of the optimal within $18$ iterations. 

For the larger topology, the mechanism converges to a rate smaller than the optimal. The reason is as follows. In this topology, the link which is getting congested
first, or in other words, has the least residual capacity remaining at all iterations of the algorithm, is a high-loss link\footnote{A link which suffers from a 
large number of physical layer losses without including losses due to collisions is referred to as high-loss link.} (loss rate $> 40\%$). 
Hence, the iteration duration has to be larger than $200$ packets to obtain an accurate estimate. 
If we increase the iteration duration to $500$ packets, as shown in Figure~\ref{fig:plot_rndl500}, the mechanism converges to within $5\%$ of the optimal 
within $15$ iterations. Note that, in general, routing schemes like ETX~\cite{etx} will avoid the use of such high-loss links for routing, and hence,
for most cases, an iteration duration of $200$ packet suffices.

\begin{figure}[htb]
 \centering
\includegraphics[width=5.0cm]{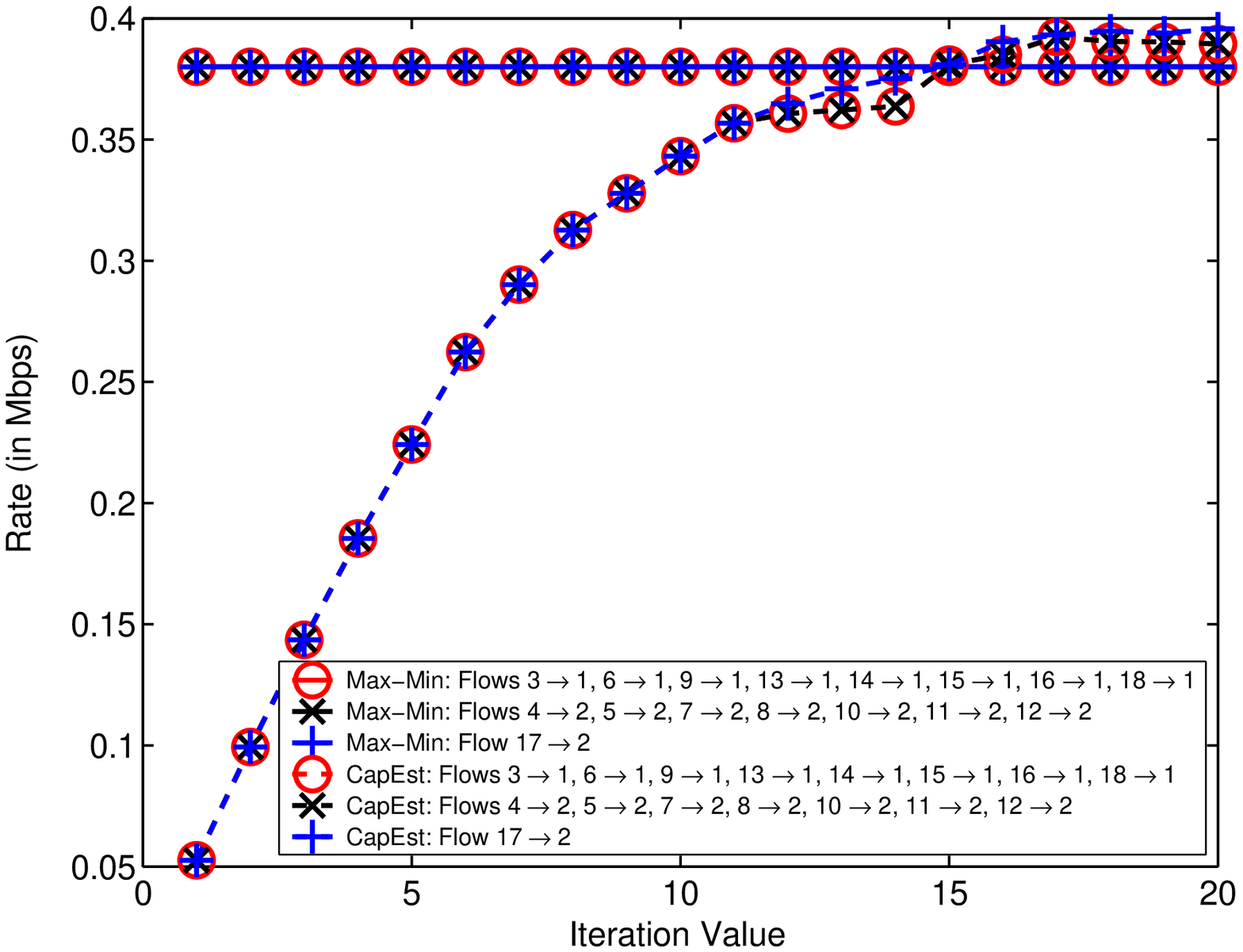}
  \caption{Performance of CapEst: Deployment at Houston.}
\label{fig:plot_houston}
\end{figure}

\subsection{A Real Topology: Deployment at Houston}
The next topology is derived from an outdoor residential deployment in a Houston neighborhood~\cite{ed_knightly_mesh_experiments}.
The node locations (shown in Figure~\ref{fig:trice}) are derived from the deployment and fed into the simulator. The physical channel that we use 
in the simulator is a two-ray path loss model with log-normal shadowing and Rayleigh fading. The ETX routing metric~\cite{etx} (based on
data loss in absence of collisions) is used to set up the routes. Nodes $1$ and $2$ are connected to the wired world and serve as gateways for this deployment. 
All other nodes route their packets toward one of these nodes (whichever is closer in terms of the ETX metric). The resulting topology
as well as the routing tree is also shown in Figure~\ref{fig:trice}.

Figure~\ref{fig:plot_houston} plots the evolution of the rate assigned to each flow by the centralized rate allocator.
Again, the mechanism converges to within $5\%$ of the optimal in less than $12$ iterations. 

\begin{figure}
\centerline{\subfigure[]{\includegraphics[width=4.0cm]{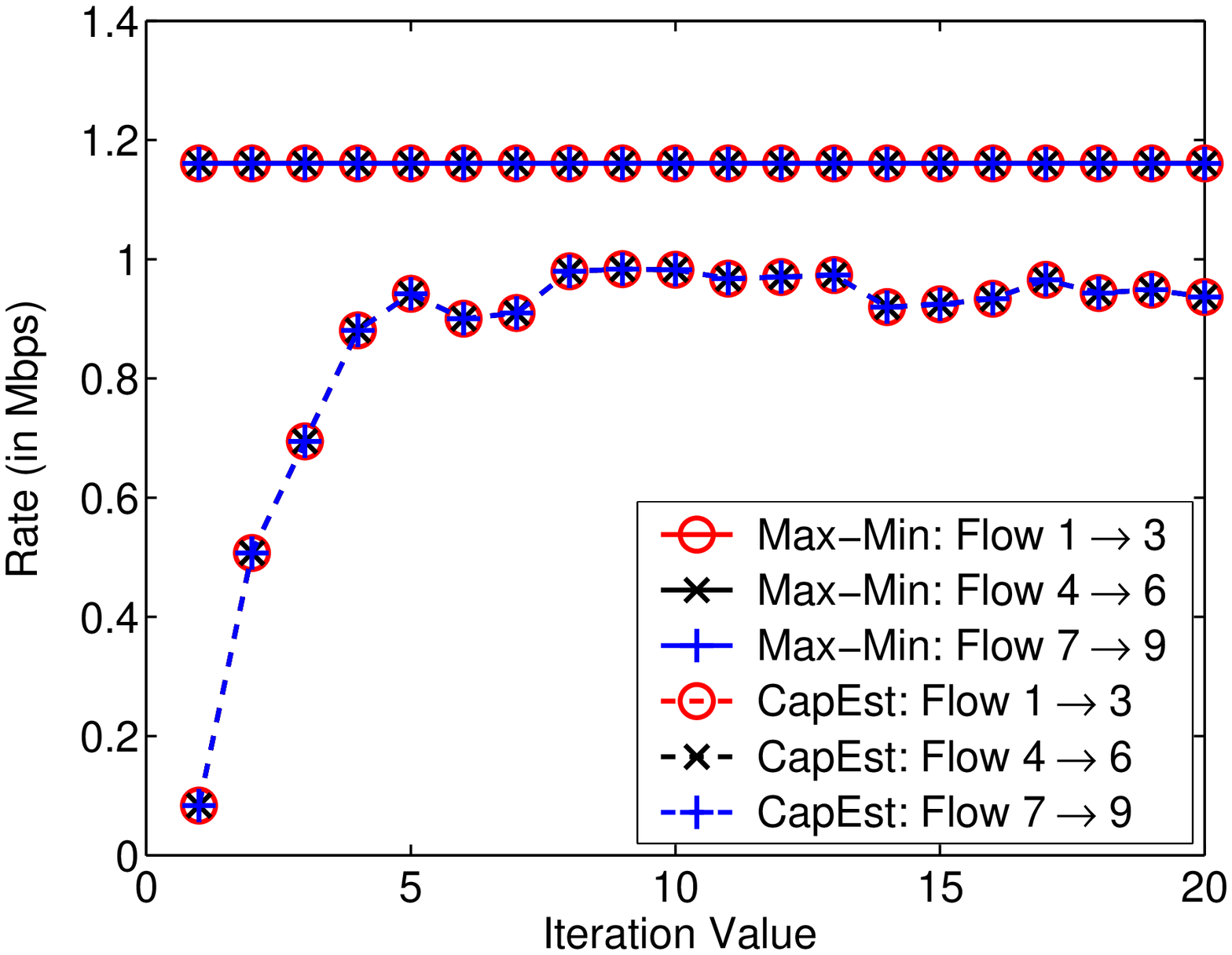}
\label{fig:plot_75_fim}}
\hfil
\subfigure[]{\includegraphics[width=4.0cm]{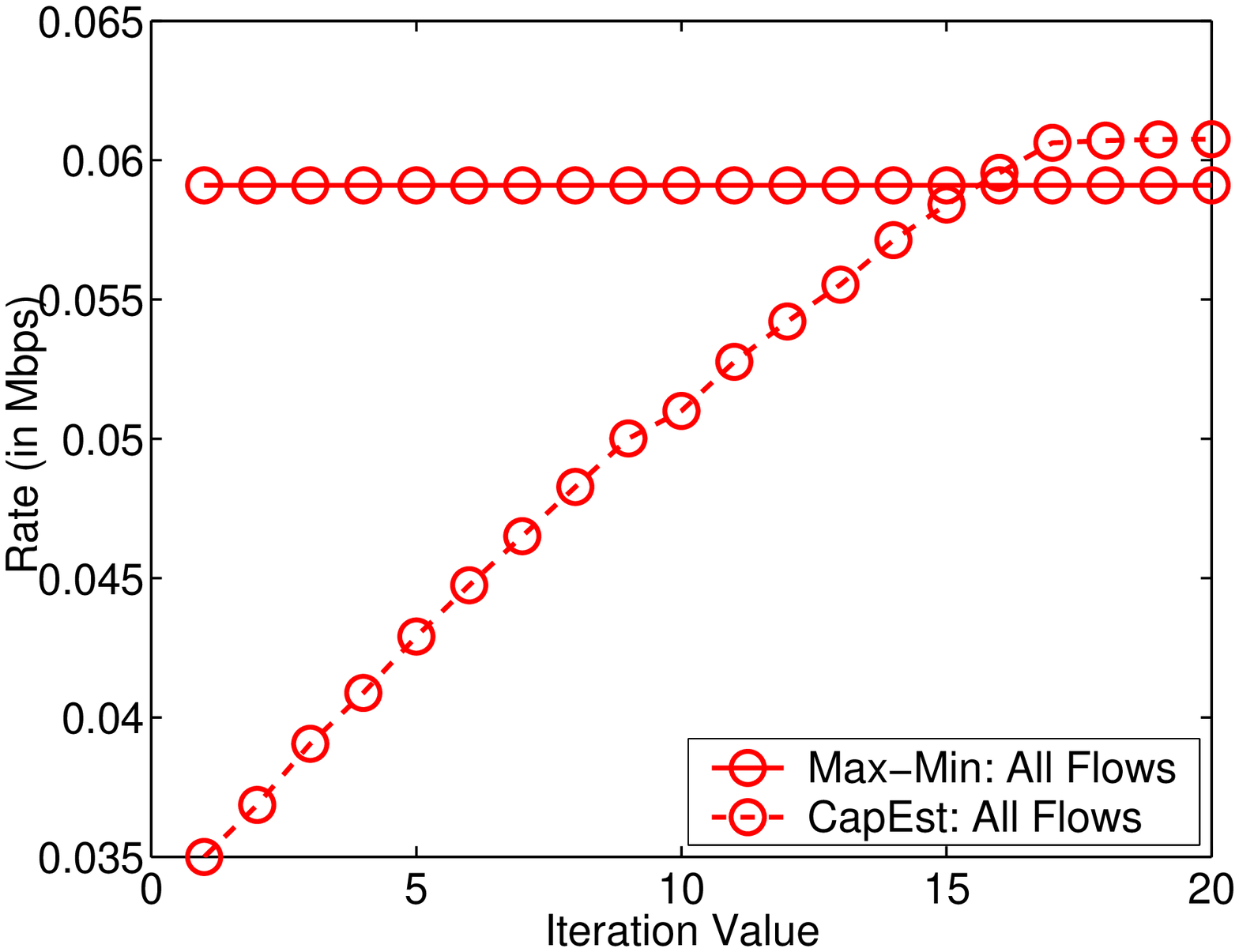}
\label{fig:plot_rndl500}}}
\caption{Performance of CapEst with a different iteration duration. (a) Flow in the Middle with iteration duration = $75$ packets. 
(b)  Random topology: $100$ nodes, $10$ flows with iteration duration = $500$ packets.}
\end{figure}

\subsection{Impact of a Smaller Iteration Duration}
\label{sec:sim_interval_length}
In this section, we evaluate the impact of using a smaller iteration duration on performance. We plot the evolution of the rate assigned to each flow 
in Figure~\ref{fig:plot_75_fim} for the flow in the middle topology with one iteration duration $= 75$ packets. We observe that the rate converges to a value smaller than
the optimal. (Note that we had made a similar observation for the $100$ node random topology in Section~\ref{sec:rnd}.) 
In general, for all topologies we studied, we observed that using an iteration duration not sufficiently large to allow the estimated expected
service time to converge leads to the mechanism converging to a rate smaller than the optimal, however, the mechanism still always converged and
did not suffer from oscillations. 

\begin{figure}
\centerline{\subfigure[]{\includegraphics[width=4.0cm]{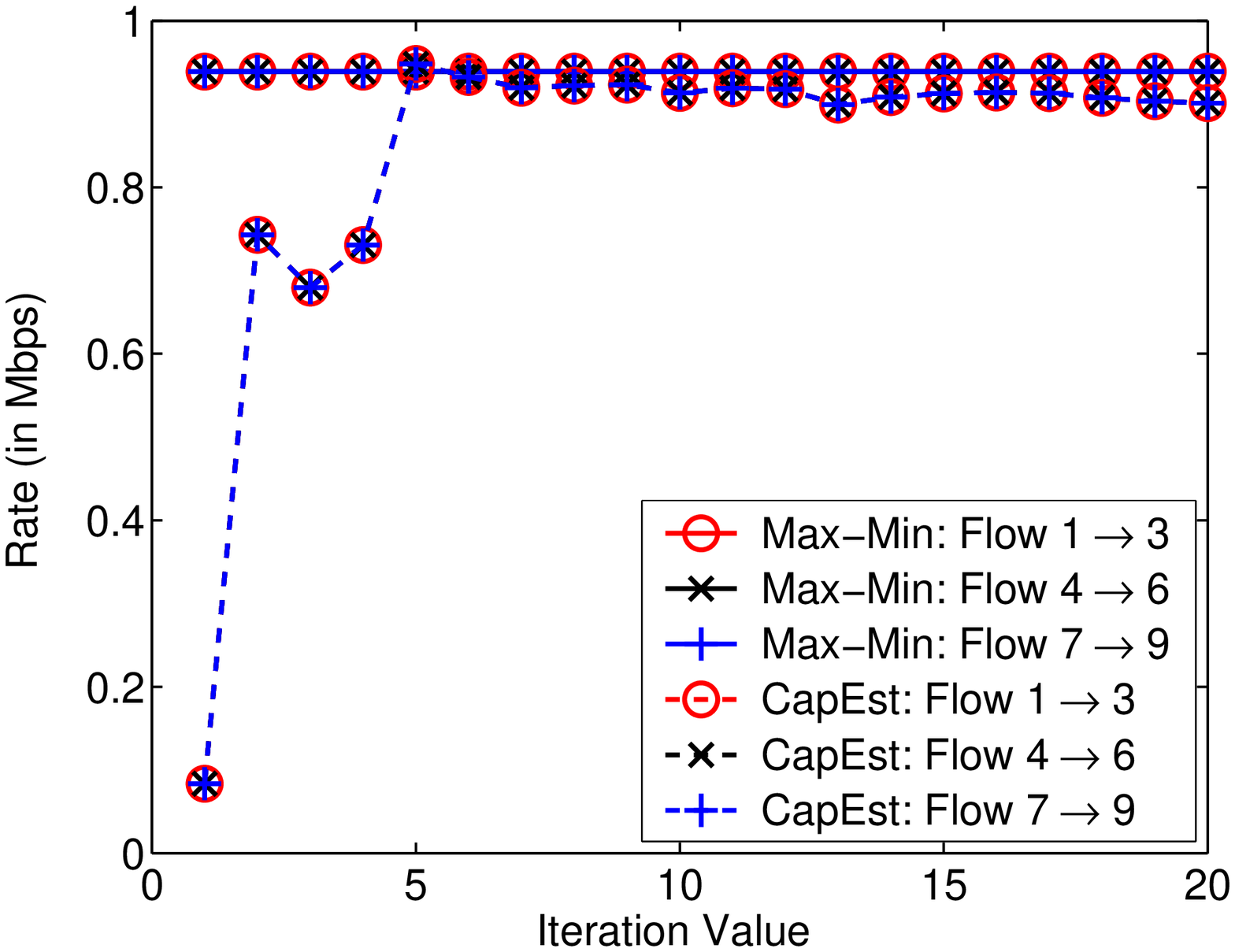}
\label{fig:plot_fim_rtscts}}
\hfil
\subfigure[]{\includegraphics[width=4.0cm]{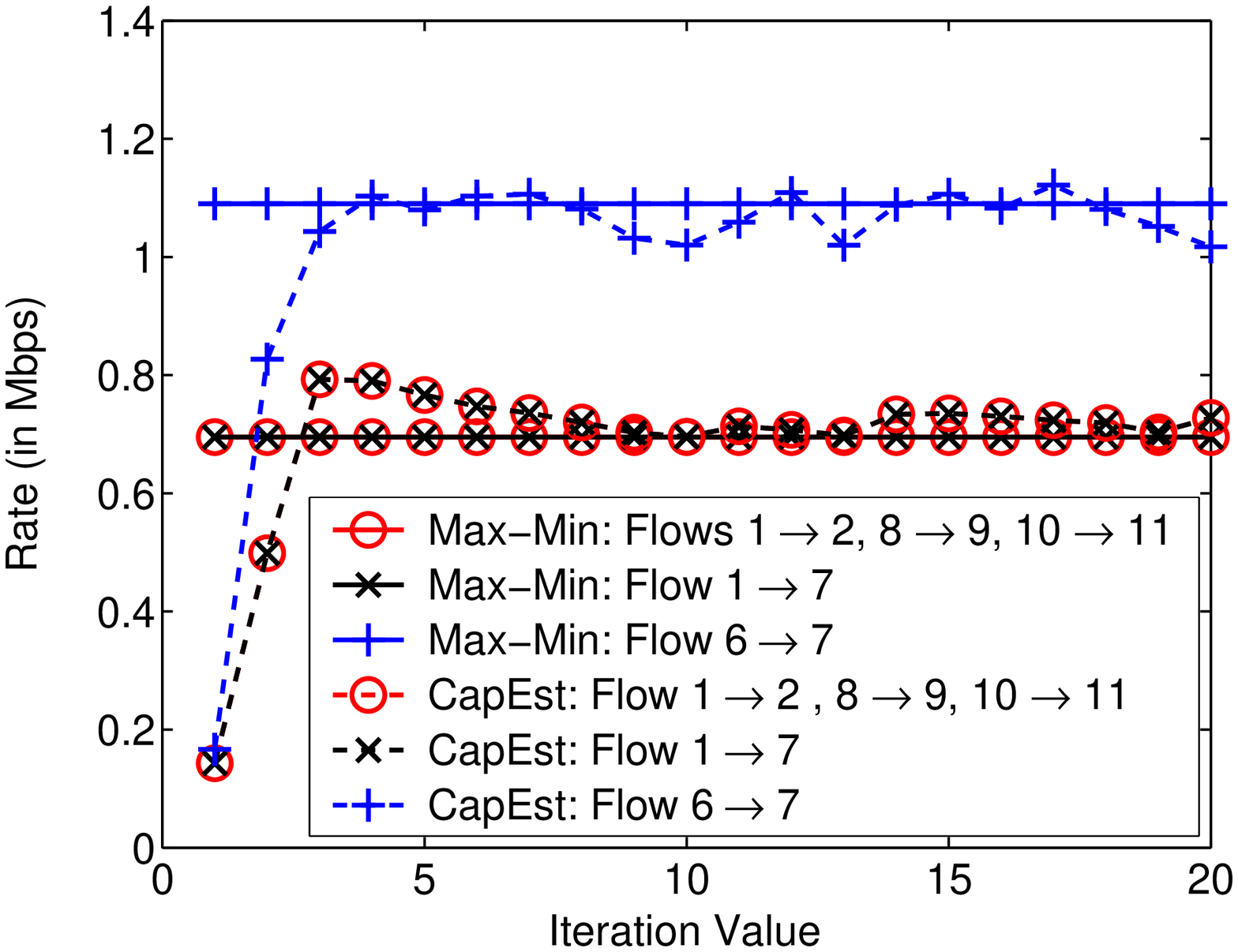}
\label{fig:plot_cc_rtscts}}}
\caption{Performance of CapEst with IEEE 802.11 with RTS/CTS. (a) Flow in the Middle. (b) Chain-cross.}
\end{figure}

\subsection{Different MAC layers}
\label{sec:diff_mac}
An attractive feature of CapEst is that it does not depend on the MAC/PHY layer being used. 
Hence, in future, if one decides to use a different medium access or physical layer, CapEst can be retained without any changes. 
In this section, we evaluate the performance of CapEst with two different medium access layers.

\begin{figure}
\centerline{\subfigure[]{\includegraphics[width=4.0cm]{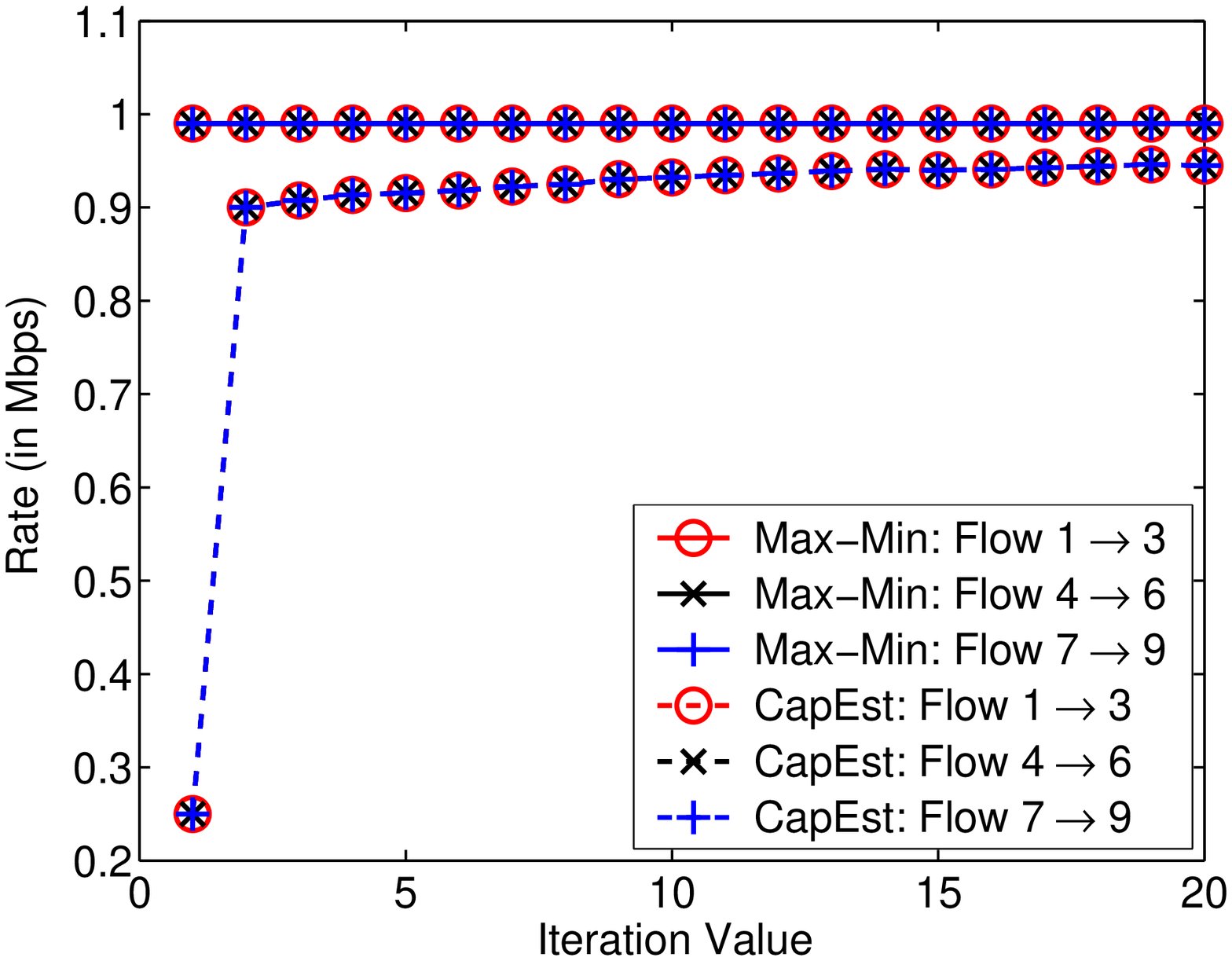}
\label{fig:plot_fim_backpressure}}
\hfil
\subfigure[]{\includegraphics[width=4.0cm]{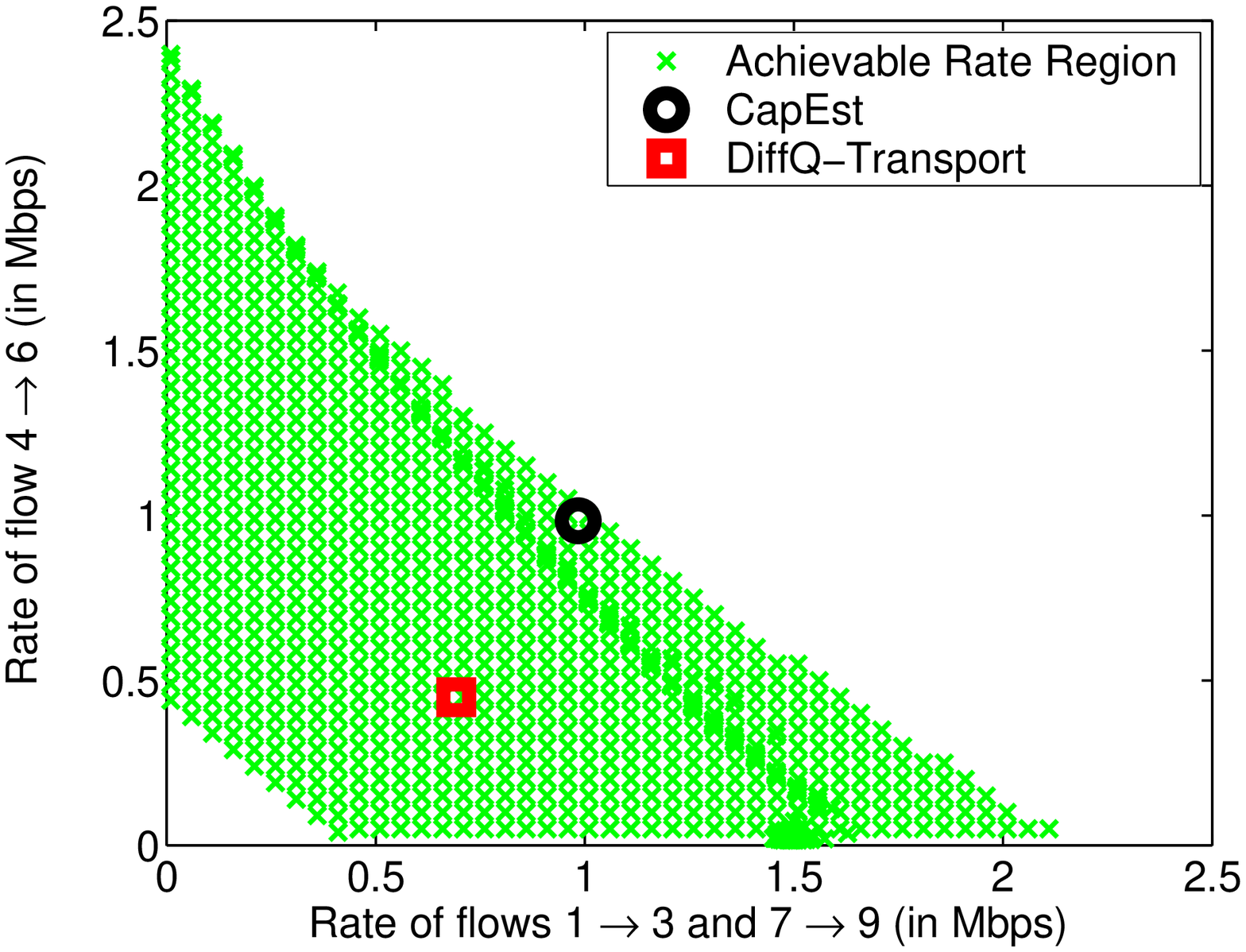}
\label{fig:ach_rate_backpressure}}}
\caption{(a) Performance of CapEst with DiffQ-MAC: Flow in the Middle. (b) CapEst vs DiffQ-Transport over DiffQ-MAC.}
\end{figure}

\subsubsection{With RTS/CTS} We first evaluate the performance of CapEst with IEEE 802.11 DCF with RTS/CTS. Figures~\ref{fig:plot_fim_rtscts} and ~\ref{fig:plot_cc_rtscts} 
plot the evolution of the assigned flow rates for the flow in the middle topology and the chain-cross topology respectively.
For both the topologies, the mechanism converges to within $5\%$ of the optimal within $6$ iterations.   

\subsubsection{Back-pressure MAC} We next evaluate the performance of CapEst with a back-pressure-based random-access protocol~\cite{warrier:infocom,srikant:random,shroff:random}.
The fundamental idea behind back-pressure based medium access is to use queue sizes as weights to determine which link gets scheduled.
Solving a max-weight formulation, even in a centralized manner, is NP-hard~\cite{sharma:mobicom}. So, multiple researchers have suggested using 
a random access protocol whose channel access probabilities inversely depend on the queue size~\cite{warrier:infocom,srikant:random,shroff:random}.
This ensures that the probability of scheduling a packet from a larger queue is higher. 
The most recent of these schemes is {\it DiffQ}~\cite{warrier:infocom}. DiffQ comprises of both a MAC 
protocol as well as a rate control protocol. We refer to them as DiffQ-MAC and DiffQ-Transport respectively. 
The priority of each head of line packet in a queue
is determined by using a step-wise linear function of the queue size, and each priority is mapped to a different AIFS, CWMin and CWMax parameter
in IEEE 802.11(e). 

Back-pressure medium access is very different from the traditional IEEE 802.11 DCF in conception. However, CapEst can still accurately
measure the capacity at each edge. Figure~\ref{fig:plot_fim_backpressure} plots the evolution of the assigned flow rates for the flow in the middle topology
with DiffQ-MAC. The different priority levels as well as the AIFS, CWMin and CWMax values
being used are the same as the ones used in~\cite{warrier:infocom}. Again, the mechanism converges to the optimal values within $5\%$ of the optimal
within $12$ iterations. 

This set-up also demonstrates the advantage of having a rate control mechanism which does not depend on the MAC/PHY layers. 
Optimal rate-control protocols for a scheduling mechanism which solves the max-weight problem at each step are known. 
However, if we use a distributed randomized scheduling mechanism like DiffQ-MAC, these rate control protocols are no longer optimal. 
But, using a rate allocation mechanism based on CapEst, which makes no
assumption on the MAC layer, ensures convergence to a rate point close to the optimal. For example, Figure~\ref{fig:ach_rate_backpressure}
plots the achievable rate region (or the feasible rate region) for DiffQ-MAC for the flow in the middle topology.
We plot the rate of the middle flow against the rate of the outer two flows. (Using symmetry to assume that the rate of the outer flows is equal
simplifies the figure as it becomes a two-dimensional figure instead of three-dimensional one.) 
Figure~\ref{fig:ach_rate_backpressure} also plots the throughput achieved by CapEst after $15$ iterations of the algorithm as well as the throughput achieved
by DiffQ-Transport (originally proposed by~\cite{low:aimd_wireless} and shown to be optimal 
with centralized max-weight scheduling). The figure shows that CapEst allocates throughput within $5\%$ of the optimal
while DiffQ-Transport achieves only $55\%$ of the optimal throughput.

%TechReport
\begin{figure}
\centerline{\subfigure[]{\includegraphics[width=4.0cm]{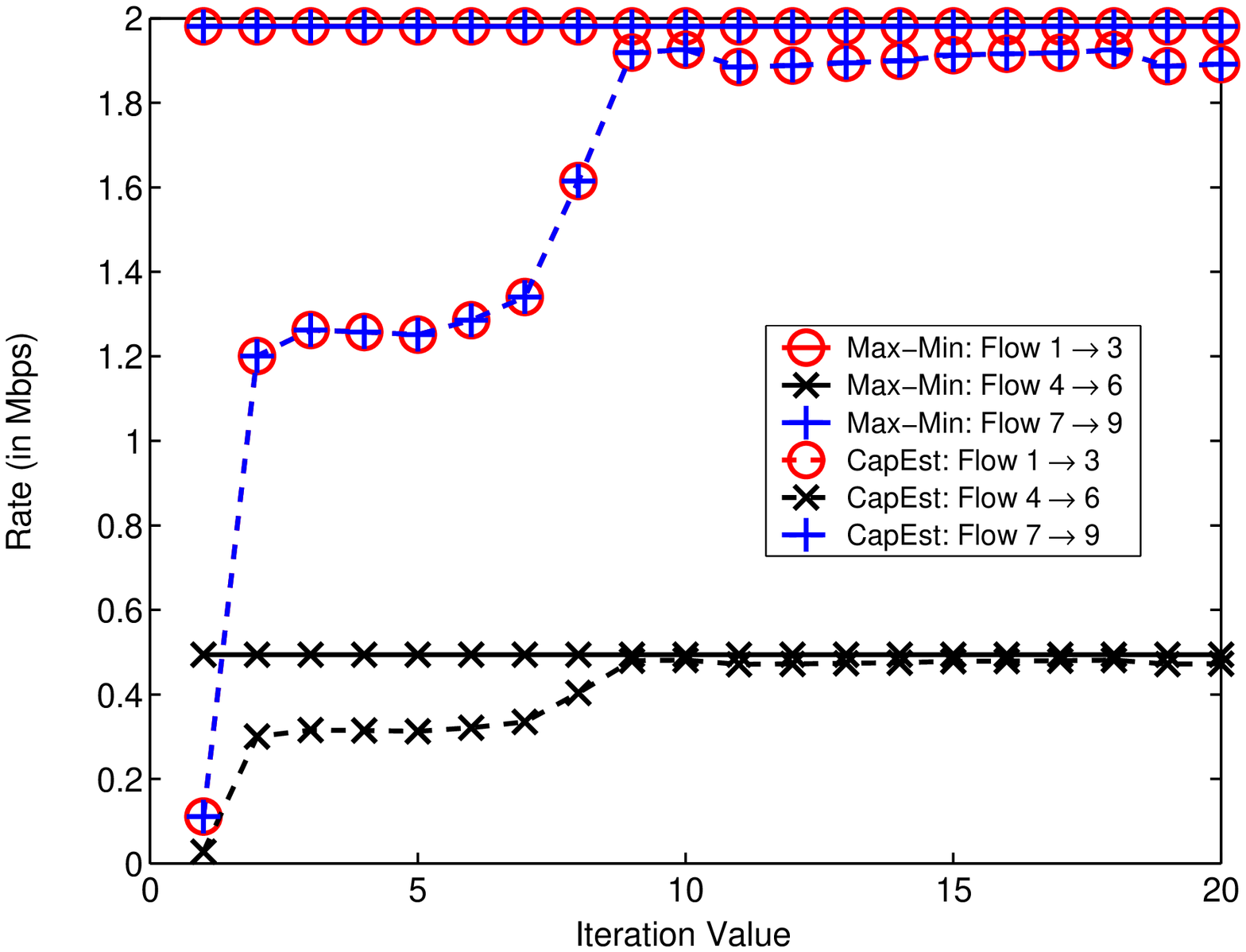}
\label{fig:plot_fim_weighted}}
\hfil
\subfigure[]{\includegraphics[width=4.0cm]{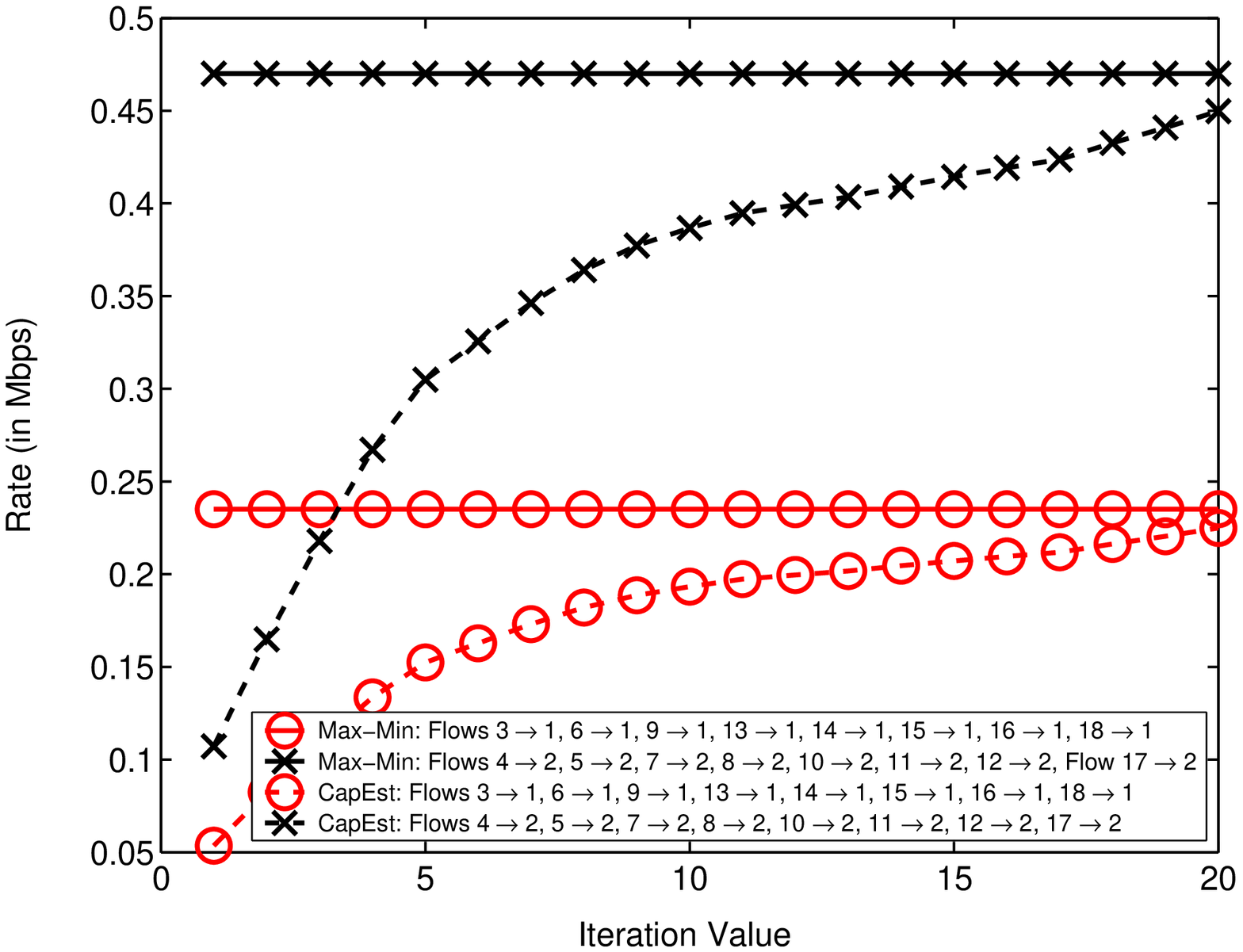}
\label{fig:plot_houston_weighted}}}
\caption{Performance of CapEst with weighted fairness. (a) Flow in the Middle. (b) Deployment at Houston.}
\end{figure}

\subsection{Weighted Fairness}
\label{sec:weighted}
Section~\ref{sec:max-min-calc} describes how to use CapEst to obtain max-min fairness. However, the same mechanism can easily be adapted to obtain weighted fairness,
which is defined as follows. Let $w_f > 0, \forall f \in \mathcal{F}$ denote the weight of a flow. Then, weighted fairness will allocate rates
$r_f, \forall f \in \mathcal{F}$ such that $r_{f_1}/r_{f_2} = w_{f_1}/w_{f_2}, \forall f_1, f_2 \in F$, the allocated flow rates are feasible, 
and the rate of no flow can be increased without reducing the rate of any other flow. 

Given $w_f, \forall f \in \mathcal{F}$, we now describe a modification of the methodology proposed in Section~\ref{sec:max-min-calc} to achieve weighted fairness. 
The central allocator will now update the value of $r_{i \rightarrow j}^{max}$ according to the following equation: 
$r_{i \rightarrow j}^{max} = r_{i \rightarrow j}^{allocate} + \frac{1/E\left[ S_{i \rightarrow j} \right] - \lambda_{i \rightarrow j}}
{\sum_{k \rightarrow l \in N_{i \rightarrow j}} \sum_{f \in \mathcal{F}} w_f I(f, k \rightarrow l)}$.
The equation to update the value of $r_{i \rightarrow j}^{allocate}$ remains the same,
and the new flow rates are updated as $r_f^{new} = \mbox{min}_{i \rightarrow j \in P_f} w_f r_{i \rightarrow j}^{allocate}$.

Figures~\ref{fig:plot_fim_weighted} and~\ref{fig:plot_houston_weighted} plot the evolution of flow rates with CapEst for the flow in the middle topology
and the deployment at Houston respectively where $w_{1 \rightarrow 3} = w_{7 \rightarrow 9} = 4$ and $w_{4 \rightarrow 6} = 1$ 
for the flow in the middle topology, while for the deployment at Houston, $w_{f} = 1$ for all flows $f$ being routed towards gateway $1$
and $w_f = 2$ for all flows $f$ being routed towards gateway $2$. CapEst converges to within $5\%$ of the optimal within $18$ iterations of the algorithm.

\begin{figure}
\centerline{\subfigure[]{\includegraphics[width=4.0cm]{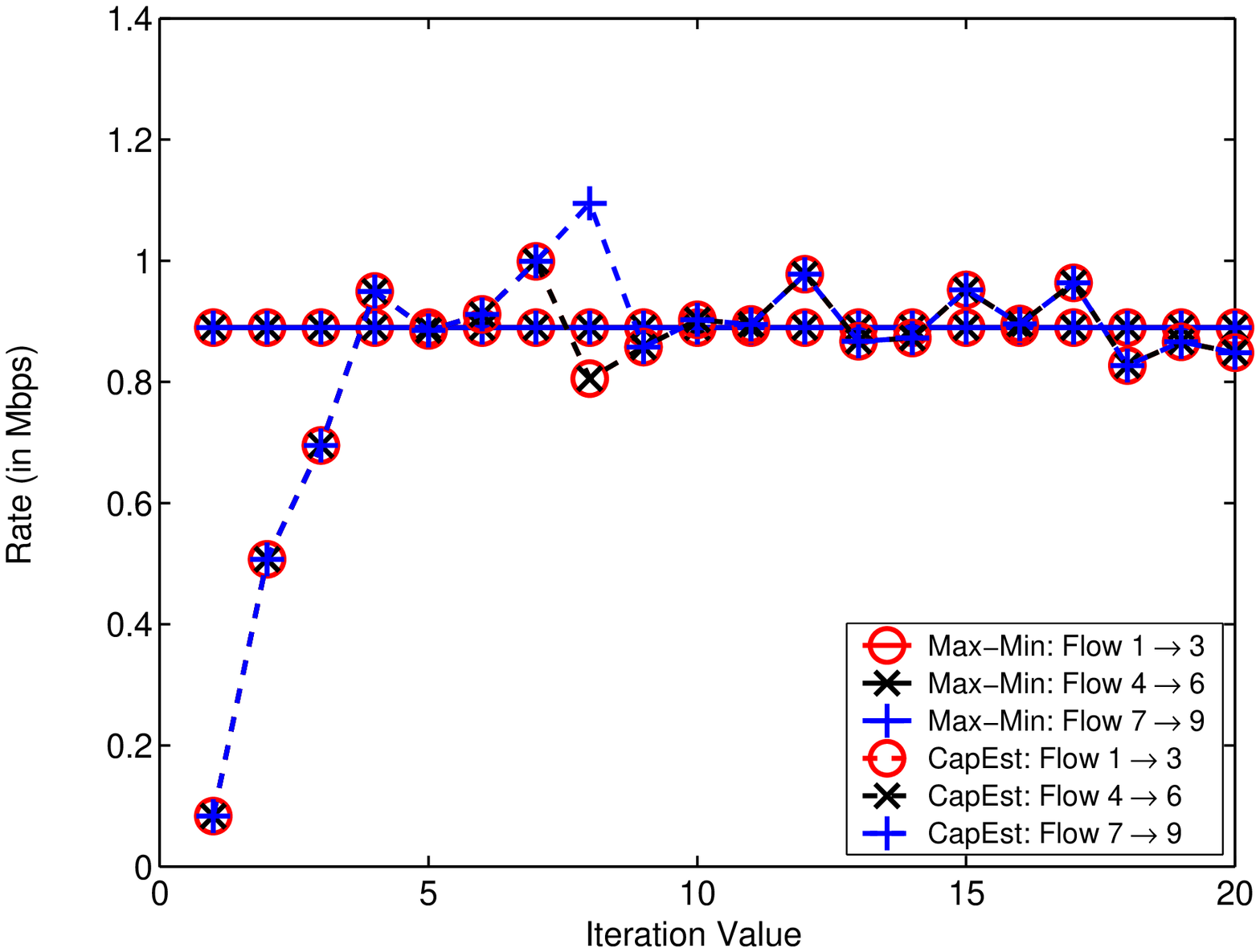}
\label{fig:plot_fim_finite}}
\hfil
\subfigure[]{\includegraphics[width=4.0cm]{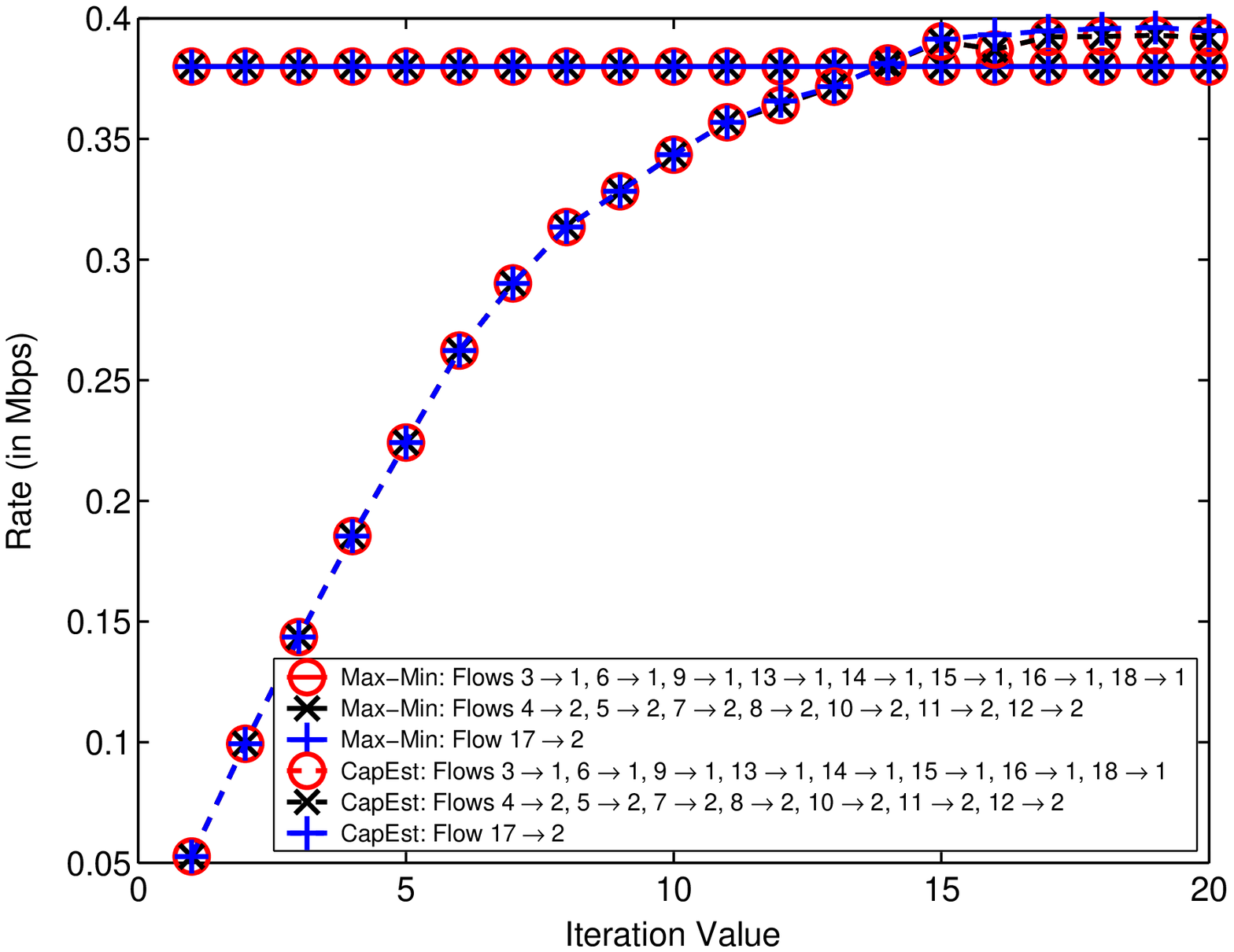}
\label{fig:plot_houston_finite}}}
\caption{Performance of CapEst with finite MAC retransmit limits. (a) Flow in the Middle. (b) Deployment at Houston.}
\end{figure}

\subsection{Finite Retransmit Limits}
\label{sec:sim_finite}
In this section, we set the IEEE 802.11 retransmit limits to their default values. Thus, packets may be dropped at the MAC layer.
We use the methodology proposed in Section~\ref{sec:finite} to update the estimate of the expected service time for lost packets. 
Figures~\ref{fig:plot_fim_finite}, and~\ref{fig:plot_houston_finite} show the evolution of allocated rates 
for the flow in the middle topology and the deployment at Houston respectively. 
There is slightly more variation in the allocated rates, however, not only does CapEst converge but also the convergence time
remains the same as before. We evaluate CapEst for all the other scenarios described above with the default retransmit values for IEEE 802.11,
and our observations remain the same as before.  

\begin{figure}
\centerline{\subfigure[]{\includegraphics[width=4.0cm]{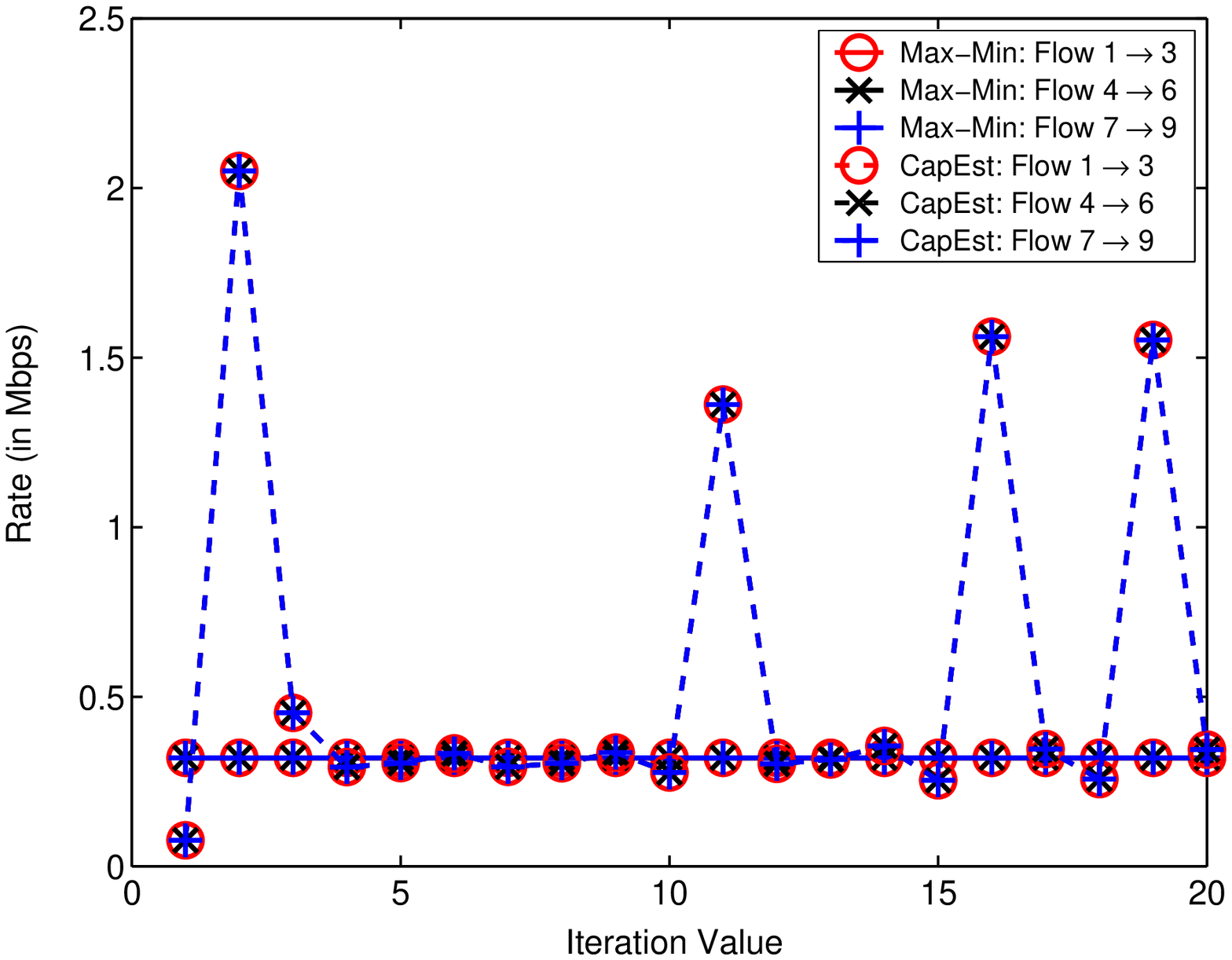}
\label{fig:auto1}}
\hfil
\subfigure[]{\includegraphics[width=4.0cm]{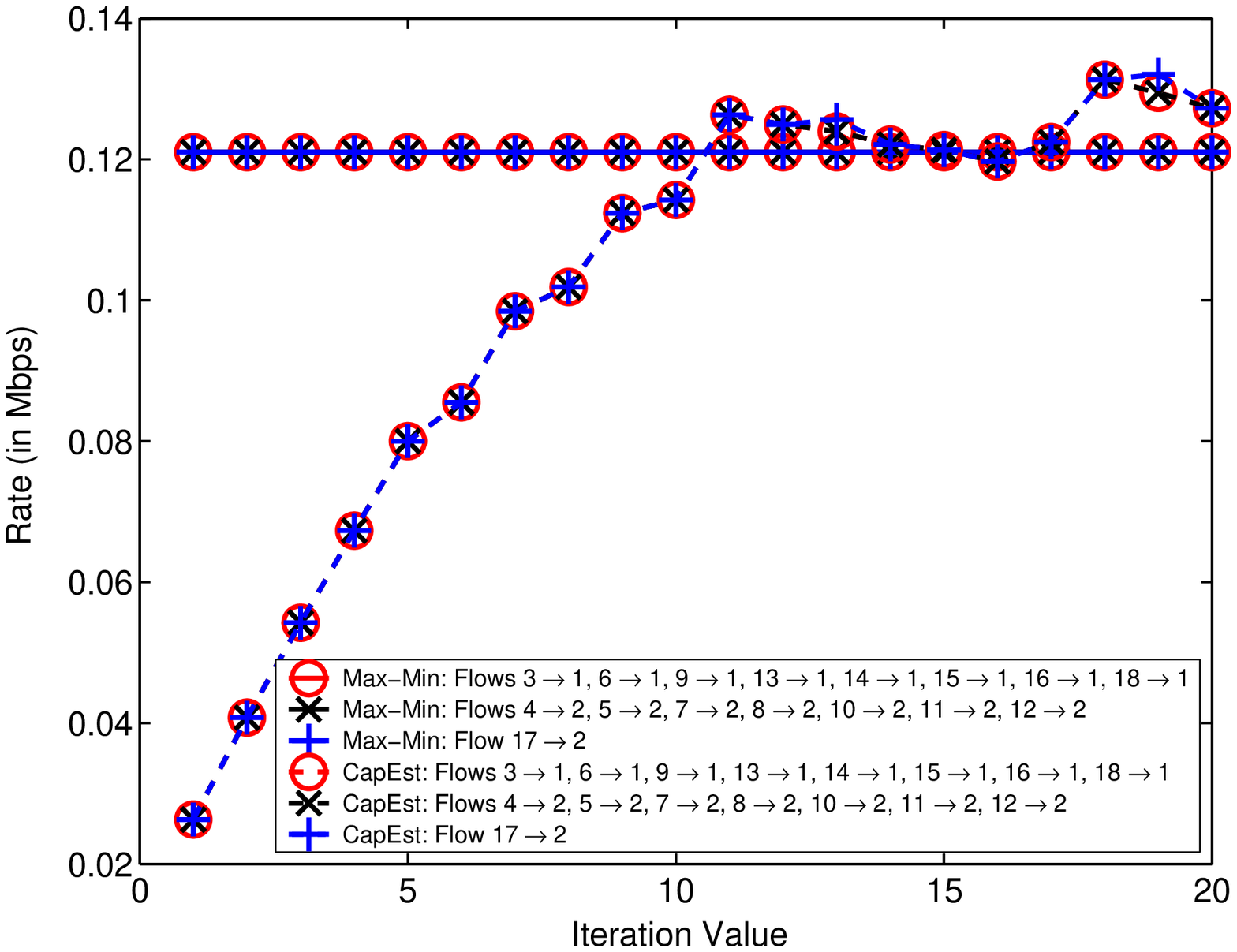}
\label{fig:auto2}}}
\caption{Performance of CapEst with auto-rate adaptation. (a) Flow in the Middle. (b) Deployment at Houston.}
\end{figure}

\subsection{Auto-Rate Adaptation}
\label{sec:adapt}
In this section, we switch on the default auto-rate fallback mechanism of Qualnet. 
We use the methodology proposed in Section~\ref{sec:auto-rate} to constrain the rate updates. 
Figures~\ref{fig:auto1} and~\ref{fig:auto2} show the evolution of allocated rates for the flow in the middle topology 
and the deployment at Houston respectively. Again, not only does CapEst converge to the correct rate allocation
but also the convergence time remains the same as before. However, the variation in the rates being assigned is larger than before. 
This is due to the fact that at lower rates, fewer collisions are observed, hence the data rate at each link is higher. 
Thus, the residual capacity estimate observed will be larger and hence, the variation in rates allocated will be larger.
In other words, the larger variation in rates is due to a larger variation in data rates caused by auto-rate fallback and not CapEst.
Also, this variation is less prominent in the real scenario, the deployment at Houston, where due to fading, there is significant channel losses in absence of collisions and
the expected service time estimate varies less due to collisions.
We evaluate CapEst for all the other scenarios described above with the default auto-rate fallback mechanism of Qualnet, and our observation remains the same.
 
\begin{figure}
\centerline{\subfigure[]{\includegraphics[width=4.0cm]{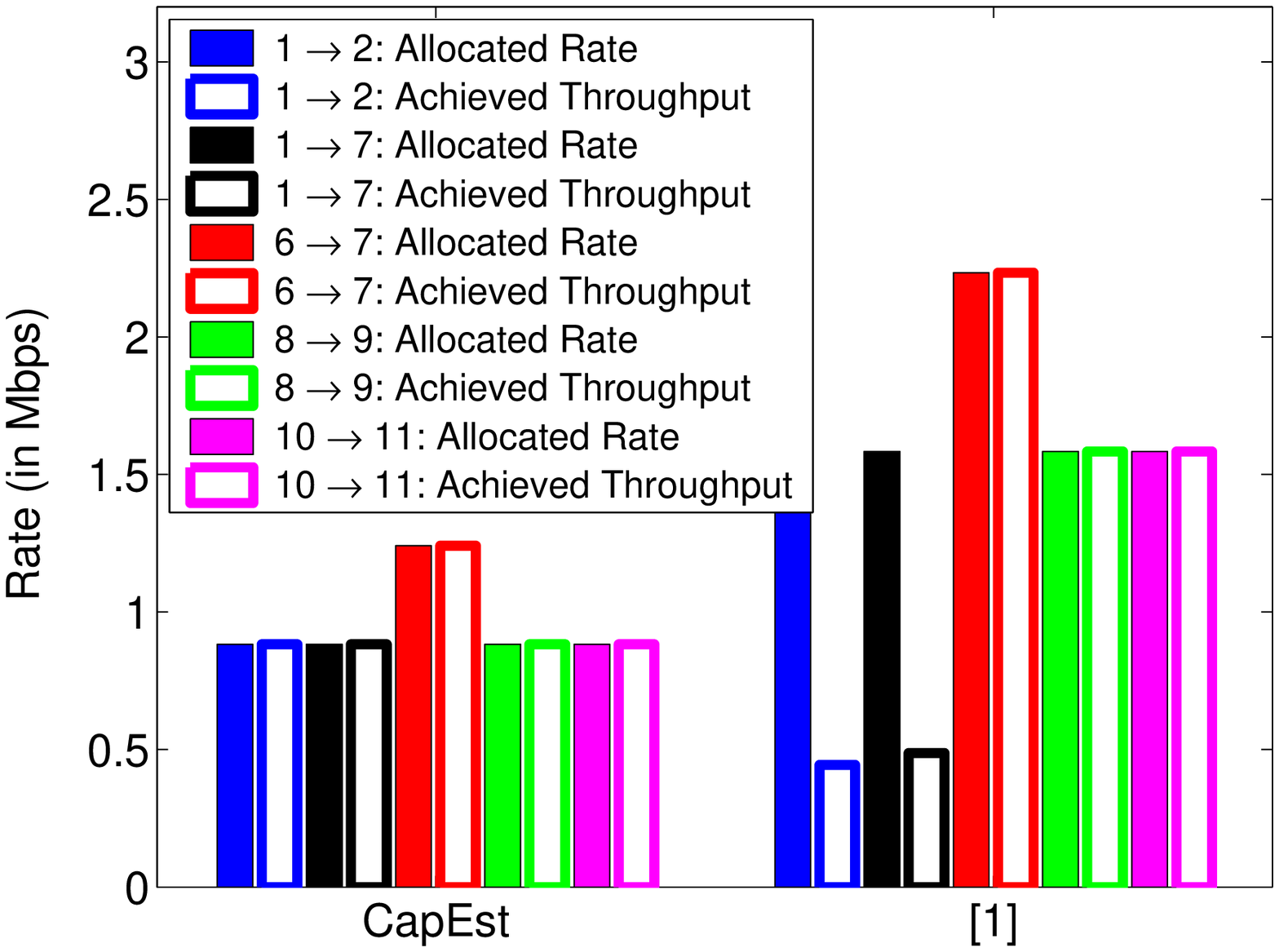}
\label{fig:sigcomm08_compare}}
\hfil
\subfigure[]{\includegraphics[width=4.0cm]{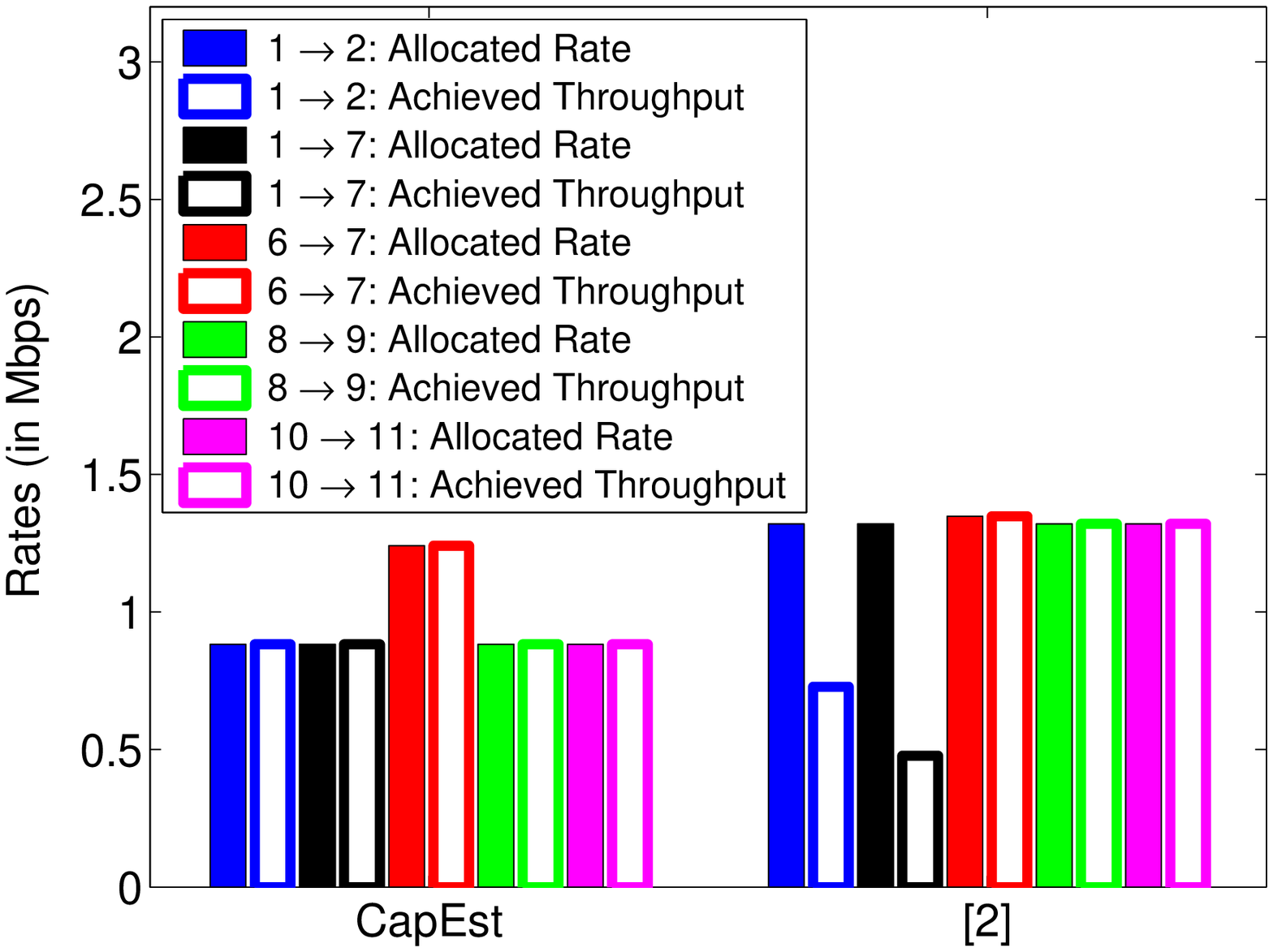}
\label{fig:conext_compare}}}
\caption{(a) CapEst vs~\cite{qiu:sigcomm}. (b) CapEst vs~\cite{ramesh:conext}.}
\end{figure}

\subsection{Comparison with Prior Works}
Two prior works have discussed centrally allocating max-min fair rates for a given topology~\cite{qiu:sigcomm,ramesh:conext}.
~\cite{qiu:sigcomm} assumes that the entire topology is precisely known, and then sets up a convex optimization problem based on an approximate model
for IEEE 802.11 DCF without RTS/CTS to determine the exact flow rates.~\cite{ramesh:conext} assumes that only which node interferes with
whom is known, and then sets up a linear program based on an approximate model for IEEE 802.11 DCF to determine the exact flow rates. 
Note that the model used by~\cite{qiu:sigcomm} is more complex as well as more accurate than~\cite{ramesh:conext}.
Both these methods are model-based, and hence suffer from all the shortcomings any model-based capacity estimation technique suffers from. 
Additionally, since they are based on approximate models, they will yield either an over-estimate or an under-estimate of the actual rates,
and as reported by~\cite{qiu:sigcomm,ramesh:conext}, this over or under estimation can be more than $40\%$ of the actual value. 
%And since these methods are not iterative, there is no way to correct these incorrect rates.

Figure~\ref{fig:sigcomm08_compare} compares the rates assigned as well as the actual throughput achieved 
by CapEst after $15$ iterations and the method proposed in ~\cite{qiu:sigcomm} (set-up to obtain max-min fairness) while Figure~\ref{fig:conext_compare} compares the same
for CapEst and the method proposed in~\cite{ramesh:conext} (set-up to obtain max-min fairness) for the chain-cross topology.
We observe that the methodologies of both the prior works overestimate capacity which leads to an infeasible rate allocation leading to a lot of packet drops and a 
much lower actual throughput value for the two flows in the middle ($1 \rightarrow 2$ and $1 \rightarrow 7$). 

Note that the initial flow rates can be assigned based on the solution of either of these two methods. However, after assigning these initial
rates, the CapEst mechanism can be used to converge to the correct flow rates. In other words, instead of starting CapEst from near-zero rates, 
its initial starting rate allocation can be determined using either one of these two methods. 

\subsection{Discussion}
\label{sec:sim_discuss}
\noindent {\bf Short Flows.} The $200$ packets in an iteration duration can belong to any flow, long-term or short-term.
Hence, CapEst does not need any additional support when numerous short flows are present in a network. 
For example, the centralized rate allocation mechanism can easily determine the rate of a new short flow based on the previous residual capacity estimate.
If we introduce $50$ short flows of $10$ packets each in the deployment at Houston topology, at randomly generated times after the $10^{th}$ iteration, in addition to the $16$ long flows,
and choose the source of each short flow randomly, we observe that each of these short flows are allocated adequate capacity and hence, they 
complete their data transfer quickly within $200$ ms. 

\noindent {\bf External Interference.} The presence of external interference merely increases the expected service time at links, which reduces 
the residual capacity estimate. Hence, unlike model-based capacity estimation techniques, CapEst does not need any additional
support or mechanisms to account for external interference. 

\section{Other Applications}
\label{sec:applications}
Accurately estimating capacity is an important tool which can be used in many different applications.
So far we have only used centralized rate allocation as an example to study the properties of CapEst. 
As discussed in the Introduction, CapEst can also be used with a number of other applications.
To illustrate how CapEst will get modified when used with these applications, in this section, we will explore how CapEst
fits into the following two applications: distributed rate allocation and interference-aware routing.

\subsection{Distributed Rate Allocation}
WCPCap is one of the recent distributed rate control algorithms for mesh networks which uses capacity estimation to allow the intermediate router nodes
to explicitly and precisely tell the source nodes the rate to transmit at~\cite{our:mobicom}.
It comprises of two parts. The first part estimates residual capacity at each link using a complex model for IEEE 802.11 with RTS/CTS in multi-hop networks, and the
second part divides this estimated capacity amongst end-to-end flows in a distributed manner to achieve max-min fairness. 
Since WCPCap uses a specific model to estimate capacity, it suffers from all the drawbacks which any model-based capacity estimation technique suffers from. 

To illustrate the utility of CapEst in distributed rate allocation, we combine CapEst and WCPCap. CapEst replaces the first part of WCPCap to estimate
capacity, while the second part which divides this estimated capacity is retained as such. We evaluate the performance of CapEst with WCPCap
using Qualnet simulations. Default parameters of IEEE 802.11 in Qualnet 
%TechReport
(summarized in Table~\ref{parameters}) 
are used. Retransmit limits are set to their default values and auto-rate adaptation is switched off\footnote{As stated earlier, model-based capacity estimation techniques, like
model-based WCPCap, become prohibitively expensive with auto-rate adaptation.}. 
The WCPCap parameters are set to their values used in~\cite{our:mobicom}. 
Figures~\ref{fig:fim_wcpcap} and~\ref{fig:cc_wcpcap} compare the performance of model-based WCPCap and CapEst-based WCPCap for the flow in the middle 
and the chain cross topologies respectively. We observe that WCPCap with CapEst is slightly fairer than the model-based WCPCap 
as model-based WCPCap is only as accurate as the underlying model. 
%TechReport
For example, for the chain-cross topology, CapEst allocates a higher rate to flows
$1 \rightarrow 2$, $1 \rightarrow 7$, $8 \rightarrow 9$ and $10 \rightarrow 11$ by reducing the rate of flow $6 \rightarrow 7$.

\begin{figure}
\centerline{\subfigure[]{\includegraphics[width=4.0cm]{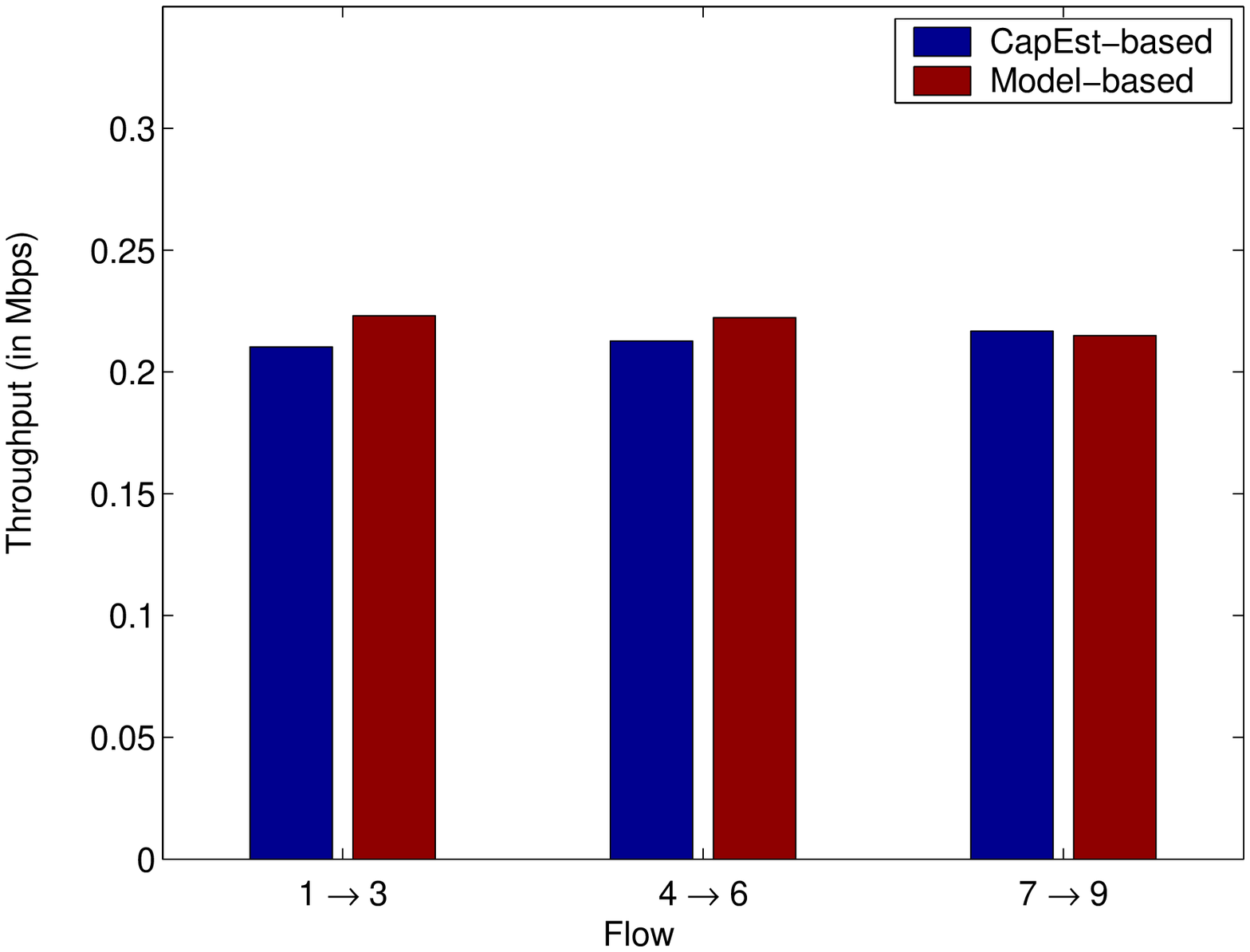}
\label{fig:fim_wcpcap}}
\hfil
\subfigure[]{\includegraphics[width=4.0cm]{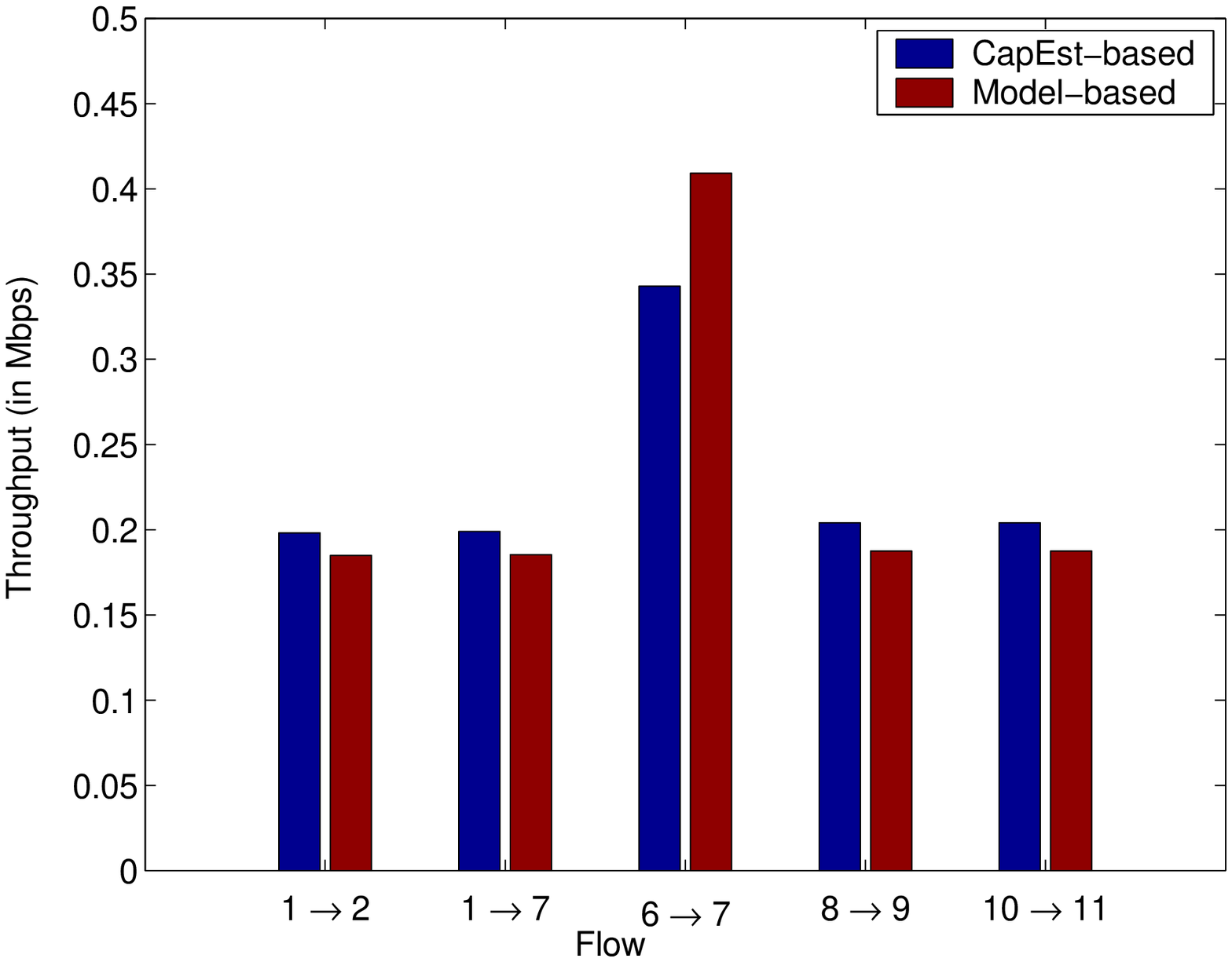}
\label{fig:cc_wcpcap}}}
\caption{Performance of CapEst implemented within WCPCap. (a) Flow in the middle topology. (b) Chain-cross topology.}
\end{figure}
%TechReport

\noindent {\bf Implementation Overhead:}
Implementation overhead for distributed rate allocation mechanisms which estimate capacity, like WCPCap, tends to be prohibitive~\cite{our:mobicom} because of
the following two reasons. (i) Complete topology information has to be collected and maintained to seed the model. 
(ii) The number of messages exchanged between links is a linear function of the number of interfering links each of the two links have. 
With CapEst, both these overheads are reduced significantly. First, as discussed earlier, the only topological information needed by a node is its neighboring nodes, 
which can be collected with very low overhead~\cite{ramesh:conext}. Second, the size of messages required is smaller, and similarly to prior works~\cite{cna:tech_report}, these messages
can be encoded in the IP header in the following fields: Type-of-Service, Identification, Flags and Fragmentation Offset. 
(This strategy is based on the observation that the amount of fragmented IP traffic is negligible~\cite{fragment:ton}.) 

\subsection{Interference-aware Routing}
\label{sec:routing}

To determine routes offering maximum capacity, Gao \etal~\cite{sigmetrics:802.11} discussed a model-based approach.
Given a path in a multi-hop network, the end-to-end throughput is determined by the minimum link capacity of the path.
So if one can figure out each link's capacity in a given path, then the minimum one is the end-to-end throughput capacity of this path. 
And amongst multiple paths between a source and a destination, the path which offers the maximum throughput is the best route.  
%Note that Gao \etal~\cite{sigmetrics:802.11}'s approach to determine link capacity suffers from all the drawbacks of any model-based capacity estimation technique. 

Since CapEst can be used to determine the capacity at each link, it can easily be used to determine the best route.
We use the same topology as used in~\cite{sigmetrics:802.11} (Figure~\ref{fig:routing_top}) to evaluate CapEst. 
Flow 2 (between $1$ and $5$) has two candidate paths. Path 1 is $1 \rightarrow 2 \rightarrow 3 \rightarrow 4 \rightarrow 5$, 
and path 2 is $1 \rightarrow 6 \rightarrow 7 \rightarrow 8 \rightarrow 9 \rightarrow 10 \rightarrow 5$. Assuming default IEEE 802.11 parameters, 
%TechReport
summarized in Table~\ref{parameters}, 
data rate of $11$ Mbps, packet size of $1024$ bytes and the rate of existing flow $1$ (between $11$ and $12$) to be $3$ Mbps, 
we evaluate the capacity offered by both paths by sending $200$ packets back-to-back along each path, to be $0.36$ Mbps and $1.28$ Mbps respectively.
Thus, path $2$ is a much better route even though it has more hops than path $1$.
Using CapEst, we arrive at the same conclusion as~\cite{sigmetrics:802.11} with a much smaller overhead and without requiring any complex computations
or any assumptions on the MAC/PHY layer being used.

%% file: conclusions.tex
\section{Conclusions}
\label{sec:conclusions}

In this paper, we propose CapEst, a mechanism to estimate link capacity in a wireless network. 
CapEst yields accurate estimates while being very easy to implement and does not require any complex computations.
CapEst is measurement-based and model-independent, hence, works for any MAC/PHY layer. CapEst can be easily modified to work with any application
which requires an estimate of link capacity. Also, the implementation overhead of CapEst is small and it does not lose accuracy when used with auto-rate
adaptation and finite MAC retransmit limits. Finally, CapEst requires no support from the underlying chipset and can be completely implemented
at the network layer.

%% file: appendix.tex
\appendix
In this appendix, we prove Theorem~\ref{thm:convergence}. We use the small buffer model for the single-hop WLAN for homogeneous non-saturated nodes proposed 
in Zhao \etal~\cite{zhao:fixed_point}. The following set of equations describe the model.
\begin{eqnarray}
\label{eqn1} & & \beta = \lambda \sigma + \lambda \gamma \left( \sigma + \frac{n}{n-1} T_s \right)  \\
\label{eqn2} & & \gamma = 1 - \left( 1 - \beta \right)^{n-1} \\
\label{eqn3} & & S = \left( b_0 + b_1 \gamma + \ldots b_k \gamma^k \right) \left( \sigma + \frac{n}{n-1} \gamma T_s \right),  
\end{eqnarray}
where $\beta$ denotes the attempt rate per slot for each node, $\lambda$ denotes the arrival rate at each node, $\sigma$ denotes the slot duration,
$\gamma$ denotes the collision probability, $n$ denotes the number of nodes, $T_s$ denotes the average packet transmission time,
$S$ denotes the expected service time, $b_0, b_1, \ldots b_k$ denotes the expected back-off window size and $k$ is the retransmission limit.

In addition to the equations governing the model, we also have the following equation. 
\begin{equation} \label{eqn:lambda} \lambda = \Psi(\lambda) = 1/S.
\end{equation}

This set of equations is a fixed point formulation; and CapEst iteratively solves this fixed point formulation. 
To prove Theorem~\ref{thm:convergence}, we will first prove that for the given formulation, a fixed point exists and is unique.
We will then show that the CapEst mechanism iteratively converges to this fixed point.
(Note that the main difference between the analysis we present and the analysis presented in~\cite{zhao:fixed_point} (on which our proofs are based)
is that they treat $\lambda$ as a constant and do not include Equation (\ref{eqn:lambda}) in their formulation.)

\begin{theorem}
The fixed point formulation governed by Equations (\ref{eqn1})-(\ref{eqn3}) and (\ref{eqn:lambda}) always has a unique solution.
\end{theorem}
\begin{proof}
We denote the fixed point formulation represented by  Equations (\ref{eqn1})-(\ref{eqn3}) and (\ref{eqn:lambda}) as $\beta = f(\beta)$. 
Since the fixed point equations are continuous, by Brouwer fixed point theorem, the fixed point must exist. To prove the uniqueness of this fixed point,
we use the technique proposed in~\cite{zhao:fixed_point}, which proves that if $f(\beta)$ is monotonically increasing and concave, the fixed point is unique.
A similar proof technique also establishes that if $f(\beta)$ is monotonically decreasing and concave, the fixed point will still be unique.
Thus to prove uniqueness, we need to prove that $f'(\beta) < 0$ and $f''(\beta) < 0$, where all derivatives are with respect to $\beta$.

(i) We first prove that $f'(\beta) < 0$. Derivating Equations (\ref{eqn2}), (\ref{eqn3}) and (\ref{eqn:lambda}) with respect to $\beta$ yields:
\begin{eqnarray}
& & \gamma' = (n-1) (1 - \beta)^{n-2} \geq 0  \\
& & S' = \gamma' \left[ \left( b_1 + \ldots k b_k \gamma^k \right) \left( \sigma + \frac{n}{n-1} T_s \gamma \right) \right. + \nonumber \\
& & \left. \left( b_0 + \ldots b_k \gamma_k\right)\frac{n}{n-1} T_s \right] \geq 0 \\
& & \lambda' =  -S'/S^2 \leq 0. 
\end{eqnarray}

Derivating Equation (\ref{eqn1}) yields $f'(\beta) = $ \newline $\frac{-1}{S^2} \left[ \left( \sigma + \left( \sigma + \frac{n}{n-1} T_s \right) \gamma 
\right) S' - S \gamma' \left( \sigma + \frac{n}{n-1} T_s \right) \right]$ \newline $= \frac{-1}{S^2} \left[ \sigma S' + \left( \sigma + \frac{n}{n-1} T_s \right) 
\left( \gamma S' - S \gamma' \right) \right]$. We next show that 
$\sigma S' + \left( \sigma + \frac{n}{n-1} T_s \right) $ $\left( \gamma S' - S \gamma' \right) > 0$. 
Expanding all the derivatives yields 
\begin{eqnarray}
& & d\gamma/d\beta \left[ \sigma \left( \left( 2 b_2 \gamma + \ldots 
k b_k \gamma^k \right) \left( \sigma + \frac{n}{n-1} T_s \gamma \right) + \right. \right. \nonumber \\ 
& & \left. \left( b_0 + \ldots b_k \gamma_k\right)\frac{n}{n-1} T_s \right) 
+ \left( \sigma + \frac{n}{n-1} T_s \right) \left( \left( b_2 \gamma^2  \right. \right. \nonumber \\ 
& & \left. + \ldots (K-1) b_k \gamma^k \right) \left( \sigma + \frac{n}{n-1} T_s \gamma \right) 
+ \left( b_0 + \ldots b_k \gamma_k \right) \nonumber \\ 
& & \left. \left. \frac{n}{n-1} T_s + b_1 \sigma \frac{n}{n-1} T_s \gamma (b_1 - b_0) \sigma^2 \right) \right] > 0. \nonumber
\end{eqnarray}
Then, since $S^2$ is positive, $f'(\beta) < 0$. 

(ii) We next prove that $f''(\beta) < 0$. Derivating Equations (\ref{eqn2}), (\ref{eqn3}) and (\ref{eqn:lambda}) twice with respect to $\beta$ yields:
\begin{eqnarray}
& & \gamma'' = -\left(n-2\right) \left( n-1 \right) \left(1 - \beta \right)^{n-3} \leq 0  \\
& & S'' =  \gamma'' \left[ \left( b_1 + \ldots + k b_k \gamma^{k-1} \right) \left( \sigma + \frac{n}{n-1} T_s \gamma \right) \right. \nonumber \\ 
& & \left. + \left( b_0 + \ldots + b_k \gamma^k \right) \frac{n}{n-1} T_s \right] + \left( \gamma' \right)^{2}\left[ \left( 2 b_2 + \ldots + \right. \right. \nonumber \\ 
& & \left. k (k-1) b_k \gamma^{k-2} \right) \left( \sigma + \frac{n}{n-1} T_s \gamma \right) + 2 \left( b_1 + \ldots + \right. \nonumber \\
& & \left. \left. k b_k \gamma^{k-1} \right) \frac{n}{n-1} T_s  \right]  \\
& & \lambda'' = \frac{S S'' - 2 \left( S' \right)^2}{S^3} 
\end{eqnarray}

Derivating Equation (\ref{eqn1}) yields $f''(\beta) = \sigma \lambda'' + \left( \lambda \gamma''\right.$ \newline $\left.  
+ 2 \lambda' \gamma' + \lambda'' \gamma \right) \left( \sigma + \frac{n}{n-1} T_s \right)$. Since $\gamma \geq 0, \lambda \geq 0, \gamma' \geq 0, 
\lambda' \leq 0$ and $\gamma'' \leq 0$, $\lambda'' < 0$ will imply $f''(\beta) < 0$. 
Since $S>0$, to show that $\lambda'' < 0$, we will have to prove that $SS'' - 2 \left( S' \right)^2 < 0$.

Let $A = \left( \sigma + \frac{n}{n-1} T_s \gamma \right)$. Then,
\begin{eqnarray}
& & SS'' - 2 \left( S' \right)^2 = X_0 + \left( \gamma' \right)^2 \left[ A^2 \left( b_0 + \ldots + b_k \gamma^k \right) \right. \nonumber \\ 
& & \left( 2 b_2 + \ldots + k (k-1) b_k \gamma^{k-2} \right) + 2 A \frac{n}{n-1} T_s \left( b_1 + \ldots + \right. \nonumber \\
& & \left. k b_k \gamma^{k-1} \right) \left( b_0 + \ldots + b_k \gamma^k \right) -2 A^2 \left( b_1 + \ldots + k b_k \gamma^{k-1} \right)^2 \nonumber \\
& & \left. -4 A T_s \frac{n}{n-1} \left( b_0 + \ldots + b_k \gamma^k \right) \left( b_1 + \ldots + k b_k \gamma^{k-1} \right) \right. \nonumber \\
\label{eqn5} & & \left. + X_1 \right] 
\end{eqnarray} 
where $X_0 \leq 0$ and $X_1 \leq 0$. The coeffecient of $A^2$ in Equation (\ref{eqn5}) is $0$ as $b_k = 2^k b_0$ and the coeffecient of
$2 A T_s \frac{n}{n-1} < 0$. Hence, $SS'' - 2 \left( S' \right)^2 < 0$ which implies $f''(\beta) < 0$.
\end{proof}

We next prove that the iterative method used in CapEst converges to this unique fixed point. 
$\Psi \left( \lambda \right) = $ \newline $\left( b_0 + b_1 \gamma + \ldots + b_k \gamma^k \right)^{-1} \left( \sigma + \frac{n}{n-1} \gamma T_s \right)^{-1}$.
At iteration $k$, $k \geq 1$, if CapEst is used to get max-min fairness, then $\lambda$ is updated as follows: 
\begin{equation} 
\label{eqn:ru}
\lambda_{k+1} =  \left( 1 - \alpha \right) \Psi\left( \lambda_k \right) + \alpha \lambda_k,
\end{equation} 
where $\alpha = 1 - \frac{1}{\sum_{k \rightarrow l \in N_{i \rightarrow j}} \sum_{f \in \mathcal{F}} I(f, k \rightarrow l)}$. 
Recall that $I(f, i \rightarrow j)$ is an indicator variable which is equal to $1$ only if flow $f \in F$ passes through the link
$i \rightarrow j \in \mathcal{L}$, otherwise it is equal to $0$. Hence, for a WLAN, $0 < \alpha = 1/n < 1$.
(Note that the iterative mechanism will start from an initial value of $\lambda$ which we label $\lambda_0$.)
The following theorem proves that for any $0 < \alpha < 1$, the CapEst mechanism will converge to the unique fixed point. 

\begin{theorem}
For the fixed point formulation described by Equations (\ref{eqn1})-(\ref{eqn3}) and (\ref{eqn:ru}), and for any value of $0 < \alpha < 1$, 
if $\Psi\left( \lambda_0 \right) > \lambda_0$, then the CapEst mechanism iteratively converges to the unique fixed point.
\end{theorem}
\begin{proof}
We will first prove that $|d\Psi\left( \lambda \right)/d\lambda| \geq C_0$ where $C_0$ is a given constant whose value we will state at the end of the proof
and $|.|$ represents the absolute value. Also, let $f(\gamma) = \left( b_0 + b_1 \gamma + \ldots + b_k \gamma^k \right)$ and 
$g(\gamma) = \sigma + \frac{n}{n-1} \gamma T_s$. Then, $b_0 \leq f(\gamma) \leq \sum_{i=0}^k b_i$ and $\sigma \leq g(\gamma) \leq \sigma + \frac{n}{n-1} T_s $
as $0 \leq \gamma \leq 1$. Finally, $b_1 \leq df(\gamma)/d\gamma \leq \sum_{i=1}^k b_i$ and $dg(\gamma)/d\gamma = \frac{n}{n-1} T_s$.
\begin{eqnarray}
& & |d\Psi\left( \lambda \right)/d\lambda| = |d\gamma/d\lambda| \frac{df(\gamma)/d\gamma g(\gamma) + f(\gamma) dg(\gamma)/d\gamma}{\left(f(\gamma) g(\gamma)\right)^2} \nonumber \\
\label{eqn:gl} & & \geq |d\gamma/d\lambda| \frac{b_1 \sigma + \frac{b_0 n T_s}{n-1}}{\left( \sum_{i=0}^k b_i \right)^2 \left( \sigma + \frac{n}{n-1} T_s \right)^2}. 
\end{eqnarray}
We next lower bound the value of $|d\gamma/d\lambda|$. Now, $|d\gamma/d\lambda| = \left(n-1\right)\left(1-\beta\right)^{n-2} |d\beta/d\lambda|$
and $|d\beta/d\lambda| = \sigma + \left( \sigma + \frac{n}{n-1} T_s \right)$ \newline $\left( \gamma + \lambda |d\gamma/d\lambda| \right)$
$\geq \frac{2 \sigma + \frac{n}{n-1} T_s}{\frac{n-1}{b_0 \sigma} \left( \sigma + \frac{n}{n-1} T_s \right) - 1} = C_2$. 
Since, $\beta \leq 1/b_0$, $|d\gamma/d\lambda| \geq \left(n-1\right)\left( 1-1/b_0 \right)^{n-2} C_2 = C_1$. 
Substituting in Equation (\ref{eqn:gl}), $|d\Psi\left( \lambda \right)/d\lambda| \geq C_0 = $ \newline $C_1 \frac{b_1 \sigma + \frac{b_0 n T_s}{n-1}}
{\left( \sum_{i=0}^k b_i \right)^2 \left( \sigma + \frac{n}{n-1} T_s \right)^2}$.

We will next prove that $\Psi\left( \lambda_k \right) \geq \lambda_k$. We will use induction. Since, it is assumed that $\Psi\left( \lambda_0 \right) > \lambda_0$,
we only need to prove that if $\Psi\left( \lambda_k \right) \geq \lambda_k$, then $\Psi\left( \lambda_{k+1} \right) \geq \lambda_{k+1}$.
\begin{eqnarray}
& & \Psi\left( \lambda_{k+1} \right) - \lambda_{k+1}  \nonumber \\
& & = \Psi\left( \lambda_{k+1} \right) - \left( \left( 1 - \alpha \right) \Psi\left( \lambda_k \right) + \alpha \lambda_k \right)  \nonumber \\
& & = \Psi\left( \lambda_{k+1} \right) - \Psi\left( \lambda_{k} \right) + \alpha \left( \Psi\left( \lambda_{k} \right) - \lambda_k \right)  \nonumber \\
\label{eqn:pr3} & & \geq C_0 \left| \lambda_{k+1} - \lambda_k \right| + \alpha \left( \Psi\left( \lambda_{k} \right) - \lambda_k \right) \\
& & = \left( C_0 \left( 1 - \alpha \right) + \alpha \right) \left( \Psi\left( \lambda_{k} \right) - \lambda_k \right) \geq 0. \nonumber
\end{eqnarray}
Note that Equation (\ref{eqn:pr3}) holds because $|d\Psi\left( \lambda \right)/d\lambda| = 
\frac{\Psi\left( \lambda_{k+1} \right) - \Psi\left( \lambda_{k} \right)}{\lambda_{k+1} - \lambda_{k}} \geq C_0$. 

Hence, $\frac{1}{\left( \sum_{i=0}^{k} b_i \right) \left( \sigma + \frac{n}{n-1} T_s \right)} \leq \lambda_k \leq \lambda_{k+1} \leq \frac{1}{b_0 \sigma}$. 
Then, the bounded and non-decreasing sequence must converge to the unique fixed point. 
\end{proof}

%% file: capest_infocom2011.bbl
\begin{thebibliography}{10}
\providecommand{\url}[1]{#1}
\csname url@samestyle\endcsname
\providecommand{\newblock}{\relax}
\providecommand{\bibinfo}[2]{#2}
\providecommand{\BIBentrySTDinterwordspacing}{\spaceskip=0pt\relax}
\providecommand{\BIBentryALTinterwordstretchfactor}{4}
\providecommand{\BIBentryALTinterwordspacing}{\spaceskip=\fontdimen2\font plus
\BIBentryALTinterwordstretchfactor\fontdimen3\font minus
  \fontdimen4\font\relax}
\providecommand{\BIBforeignlanguage}[2]{{%
\expandafter\ifx\csname l@#1\endcsname\relax
\typeout{** WARNING: IEEEtran.bst: No hyphenation pattern has been}%
\typeout{** loaded for the language `#1'. Using the pattern for}%
\typeout{** the default language instead.}%
\else
\language=\csname l@#1\endcsname
\fi
#2}}
\providecommand{\BIBdecl}{\relax}
\BIBdecl

\bibitem{qiu:sigcomm}
Y.~Li, L.~Qiu, Y.~Zhang, R.~Mahajan, and E.~Rozner, ``Predictable performance
  optimization for wireless networks,'' in \emph{Proceedings of ACM SIGCOMM},
  2008.

\bibitem{ramesh:conext}
T.~Salonidis, G.~Sotiropoulos, R.~Guerin, and R.~Govindan, ``{Online
  Optimization of 802.11 Mesh Networks},'' in \emph{Proceedings of ACM CONEXT},
  2009.

\bibitem{our:mobicom}
S.~Rangwala, A.~Jindal, K.~Jang, K.~Psounis, and R.~Govindan, ``Understanding
  congestion control in multi-hop wireless mesh networks,'' in
  \emph{Proceedings of ACM MOBICOM}, 2008.

\bibitem{sigmetrics:802.11}
Y.Gao, D.Chiu, and J.~Lui, ``Determining the end-to-end throughput capacity in
  multi-hop networks: methodology and applications,'' in \emph{Proceedings of
  ACM Sigmetrics}, 2006.

\bibitem{apoorva:ton}
A.~Jindal and K.~Psounis, ``{The Achievable Rate Region of 802.11-Scheduled
  Multihop Networks},'' \emph{IEEE/ACM Transactions on Networking}, vol.~17,
  no.~4, pp. 1118--1131, 2009.

\bibitem{mobicom:complete}
L.~Qiu, Y.~Zhang, F.~Wang, M.~Han, and R.~Mahajan, ``A general model of
  wireless interference,'' in \emph{Proceedings of ACM MOBICOM}, 2007.

\bibitem{infocom:802.11}
M.~Garetto, T.~Salonidis, and E.~Knightly, ``{Modeling Per-flow Throughput and
  Capturing Starvation in CSMA Multi-hop Wireless Networks},'' in
  \emph{Proceedings of IEEE INFOCOM}, 2006.

\bibitem{Thiran:802.11}
M.~Durvy, O.~Dousse, and P.~Thiran, ``Border effects, fairness, and phase
  transition in large wireless networks,'' in \emph{Proceedings of IEEE
  INFOCOM}, 2008.

\bibitem{kashyap:802.11}
A.~Kashyap, S.~Ganguly, and S.~Das, ``A measurement-based approach to modeling
  link capacity in 802.11-based wireless networks,'' in \emph{Proceedings of
  ACM MOBICOM}, 2007.

\bibitem{reis_sigcomm06}
C.~Reis, R.~Mahajan, M.~Rodrig, D.~Wetherall, and J.~Zahorjan,
  ``Measurement-based models of delivery and interference,'' in
  \emph{Proceedings of ACM SIGCOMM}, 2006.

\bibitem{wlan:infocom}
A.~Kumar, E.~Altman, D.~Miorandi, and M.~Goyal, ``New insights from a fixed
  point analysis of single cell {IEEE 802.11} wireless {LANS},'' in
  \emph{Proceedings of IEEE INFOCOM}, 2005.

\bibitem{tobagi:infocom}
K.~Medepalli and F.~A. Tobagi, ``{Towards Performance Modeling of IEEE 802.11
  Based Wireless Networks: A Unified Framework and Its Applications},'' in
  \emph{Proceedings of IEEE INFOCOM}, 2006.

\bibitem{wang:infocom}
X.~Wang and K.~Kar, ``{Throughput Modeling and Fairness Issues in CSMA/CA based
  Ad-hoc Networks},'' in \emph{Proceedings of IEEE INFOCOM}, 2005.

\bibitem{jain:mobicom}
K.~Jain, J.~Padhye, V.~Padmanabhan, and L.~Qiu, ``Impact of interference on
  multi-hop wireless network performance,'' in \emph{Proceedings of ACM
  MOBICOM}, 2003.

\bibitem{ATP}
K.~Sundaresan, V.~Anantharaman, H.~Hsieh, and R.~Sivakumar, ``{ATP: A Reliable
  Transport Protocol for Ad Hoc Networks},'' \emph{IEEE Transactions on Mobile
  Computing}, 2005.

\bibitem{nred}
K.~Xu, M.~Gerla, L.~Qi, and Y.~Shu, ``{TCP} unfairness in ad-hoc wireless
  networks and neighborhood {RED} solution,'' in \emph{Proceedings of ACM
  MOBICOM}, 2003.

\bibitem{warrier:infocom}
A.~Warrier, S.~Janakiraman, S.~Ha, and I.~Rhee, ``{DiffQ:} practical
  differential backlog congestion control for wireless networks,'' in
  \emph{Proceedings of IEEE INFOCOM}, 2009.

\bibitem{congestion:journal}
K.~Tan, F.~Jiang, Q.~Zhang, and X.~Shen, ``Congestion control in multihop
  wireless networks,'' \emph{IEEE Transactions on Vehicular Technology},
  vol.~56, pp. 863--873, Mar. 2007.

\bibitem{stolyarinfo08}
U.~Akyol, M.~Andrews, P.~Gupta, J.~Hobby, I.~Saniee, and A.~Stolyar, ``{Joint
  scheduling and congestion control in mobile ad hoc networks},'' in
  \emph{"Proceedings of IEEE INFOCOM"}, 2008.

\bibitem{Kumar:capacity}
P.~Gupta and P.~Kumar, ``Capacity of wireless networks,'' \emph{IEEE
  Transactions on Information Theory}, vol.~46, no.~2, 2000.

\bibitem{Tse:Mobility}
M.~Grossglauser and D.~N.~C. Tse, ``Mobility increases the capacity of ad hoc
  wireless networks,'' \emph{IEEE/ACM Transactions on Networking}, vol.~10,
  no.~4, 2002.

\bibitem{vetterli:relay}
M.~Gastpar and M.~Vetterli, ``{On the capacity of wireless networks: The relay
  case},'' in \emph{Proceedings of IEEE INFOCOM}, 2002.

\bibitem{wcp:ton}
S.~Rangwala, A.~Jindal, K.~Jang, K.~Psounis, and R.~Govindan,
  ``Neighborhood-centric congestion control for multi-hop wireless mesh
  networks,'' \emph{Under submission at IEEE/ACM Transactions on Networking},
  2010.

\bibitem{Rappaport:book}
T.~Rappaport, \emph{Wireless Communications: Principles and Practice},
  2nd~ed.\hskip 1em plus 0.5em minus 0.4em\relax Prentice Hall, 1996.

\bibitem{etx}
D.~Couto, D.~Aguayo, J.~Bicket, and R.~Morris, ``A high-throughput path metric
  for multi-hop wireless routing,'' in \emph{Proceedings of ACM MOBICOM}, 2003.

\bibitem{ed_knightly_mesh_experiments}
J.~Camp, J.~Robinson, C.~Steger, and E.~Knightly, ``Measurement driven
  deployment of a two-tier urban mesh access network,'' in \emph{Proceedings of
  ACM MOBISYS}, 2006.

\bibitem{srikant:random}
A.~Gupta, X.~Lin, and R.~Srikant, ``Low-complexity distributed scheduling
  algorithms for wireless networks,'' in \emph{Proceedings of IEEE INFOCOM},
  2007.

\bibitem{shroff:random}
C.~Joo and N.~Shroff, ``Performance of random access scheduling schemes in
  multi-hop wireless networks,'' in \emph{Proceedings of IEEE INFOCOM}, 2007.

\bibitem{sharma:mobicom}
G.~Sharma, R.~Mazumdar, and N.~Shroff, ``On the complexity of scheduling in
  wireless networks,'' in \emph{Proceedings of ACM MOBICOM}, 2006.

\bibitem{low:aimd_wireless}
L.~Chen, S.~Low, M.~Chiang, and J.~Doyle, ``Cross-layer congestion control,
  routing and scheduling design in ad-hoc wireless networks,'' in
  \emph{Proceedings of IEEE INFOCOM}, 2006.

\bibitem{cna:tech_report}
K.~Jang, R.~Govindan, and K.~Psounis, ``Simple yet efficient transparent
  airtime allocation in wireless mesh networks,'' University of Southern
  California, Tech. Rep. 915, July 2010.

\bibitem{fragment:ton}
C.~Shannon, D.~Moore, and K.~Claffy, ``Beyond folklore: Observations on
  fragmented traffic,'' \emph{IEEE/ACM Transactions on Networking}, vol.~10,
  no.~6, 2002.

\bibitem{zhao:fixed_point}
Q.~Zhao, D.~Tsang, and T.~Sakurai, ``{A simple and approximate model for
  nonsaturated IEEE 802.11 DCF},'' \emph{IEEE Transactions on Mobile
  Computing}, vol.~8, no.~11, 2009.

\end{thebibliography}
